\documentclass[11pt,a4paper,english]{article}

\usepackage{amsmath}
\usepackage{amscd}
\usepackage{latexsym}
\usepackage[english]{babel}
\usepackage{amsfonts}
\usepackage{version}
\usepackage{graphicx}
\usepackage{changepage}
\usepackage[T1]{fontenc}
\usepackage[utf8]{inputenc}
\usepackage{amsthm}
\usepackage{mathtools}
\usepackage{amssymb}
\usepackage{version}
\usepackage{fancyhdr}
\usepackage{caption}
\usepackage{stmaryrd}
\usepackage{url}
\usepackage{hyperref}
\usepackage{graphicx}
\usepackage{subfigure}

\numberwithin{equation}{section}

\newcommand{\ud}{\mathrm d}

\newcommand{\tprobcondf}[1] {\tilde P \left. \left( #1 \right| \mathcal{F}_t \right)}

\newcommand{\R}{\R}

\newtheorem{theorem}{Theorem}[section]
\newtheorem{lemma}[theorem]{Lemma}
\newtheorem{defn}[theorem]{Definition}
\newtheorem{assump}[theorem]{Assumption}
\newtheorem{prop}[theorem]{Proposition}
\newtheorem{cor}[theorem]{Corollary}
\newtheorem{rem}[theorem]{Remark}

\def \H {{\mathcal H}}
\def \Hb {{\mathbb H}}
\def \G {{\mathcal G}}
\def \Gb {{\mathbb G}}
\def \F {{\mathcal F}}
\def \Fb {{\mathbb F}}
\def \P {{\mathbb P}}

\def \R {{\mathbb R}}
\def \N {{\mathbb N}}

\newcommand{\Ind}[1]{\mathbf{1}_{\left\{ #1 \right\}}}

\def \BOmega {{\mathcal B (\Omega)}}
\def \E {E}
\def \P {\mathcal{P}}

\newcommand{\mail}[1]{\href{mailto:#1}{\texttt{#1}}}

\title{Reduced-form framework under model uncertainty}
\author{Francesca Biagini\footnote{Main affiliation: Department of Mathematics, LMU Munich, Theresienstra{\ss}e, 39, 80333 Munich, Germany, Email: \mail{biagini@math.lmu.de}} \footnote{Secondary affiliation: Department of Mathematics, University of Oslo, Box 1053, Blindern, 0316, Oslo, Norway.}
\and Yinglin Zhang\footnote{Department of Mathematics, LMU Munich, Theresienstra{\ss}e, 39, 80333 Munich, Germany, Email: \mail{zhang@math.lmu.de}}}

\date{19 March 2018}

\begin{document}
\maketitle 

\begin{abstract}

	In this paper we introduce a sublinear conditional expectation with respect to a family of possibly nondominated probability measures on a \emph{progressively enlarged} filtration. In this way, we extend the classic reduced-form setting for credit and insurance markets to the case under model uncertainty, when we consider a family of priors possibly mutually singular to each other. Furthermore, we study the superhedging approach in continuous time for payment streams under model uncertainty, and establish several equivalent versions of dynamic robust superhedging duality. These results close the gap between robust framework for financial market, which is recently studied in an intensive way, and the one for credit and insurance markets, which is limited in the present literature only to some very specific cases.
\end{abstract}

\section{Introduction}
	In this paper we study the problem of defining a sublinear conditional operator with respect to a progressively enlarged filtration and a family of probability measures possibly mutually singular to each other. In this way, we are able to derive a  consistent reduced-form framework for credit and insurance markets under model uncertainty. It is well known that the reduced-form framework can be used for credit risk modeling,	for life insurance modeling and for any context where the intensity of occurrence related to a random event of particular interest is deducible from the reference information, but the occurrence itself is not. While robust framework for financial markets has been intensively studied, a corresponding analysis for credit and insurance markets is still missing. The contribution of the paper is hence manifold.
	As the main result, we extend the classic reduced-form or intensity-based framework in \cite{Bie-Rut} to the case under model uncertainty and introduce a sublinear conditional expectation on a filtration enlarged progressively by a random event, in a way consistent with the construction in \cite{Nut-Han} on the canonical space endowed with the natural filtration.
	Secondly,	we note that credit and insurance contracts are typically payment streams, hence we study here for the first time the problem of superhedging for payment streams in continuous time under model uncertainty. Several equivalent dynamic robust superhedging dualities for payment streams are provided. In view of these superhedging results, the constructed sublinear conditional expectation can be considered as a pricing operator. %Robust superhedging strategies and prices for payment streams are defined and determined. 
	
	In the existing literature for credit risk and insurance modeling there are several papers which deal with model uncertainty, but only with dominated probability family, e.g. \cite{Li-Szi}, \cite{Jea-Mat} and \cite{Cap-Bo}. 
	When a generic family of possibly mutually singular probability measures is taken into account, the main problem of the underlying stochastic analysis is the aggregation of stochastic notions defined traditionally only under one prior (e.g. conditional expectation, stochastic integral, semimartingale decomposition) into one independent of the underlying measure, see e.g. the discussion in \cite{Son-Tou-qua}. There are many independent results using different approaches, such as capacity theory, stochastic control technique etc., which have been applied to financial market modeling, see e.g. \cite{Den-Mar}, \cite{Peng-gexp}, \cite{Den-Hu}, \cite{Son-Tou}, \cite{Guo-Pan}, \cite{Hu-Peng}, \cite{Peng-Song}, \cite{Den-Mar}, \cite{Lar-Acc} and \cite{Bia-Bou}. A pathwise solution is provided in e.g. \cite{Nut-path}, \cite{Nut-Han} and \cite{Neu-Nut-non}. 
	However, the above results hold only on the canonical space endowed with the natural filtration $\Fb$ and do not allow filtrations with dependency structure. This problem is mentioned in \cite{Aks-Hou} and solved for initial enlargement of filtration. However, the case of enlargement of filtration by introducing a totally inaccessible jump with $\Fb$-adapted intensity remains an open problem. This case is particularly relevant to describe an event which occurs as a surprise but admits observable occurrence intensity under the reference filtration $\Fb$, as in the case of the default of a financial institute or the decease of a person.
The existing construction of sublinear conditional expectation on $\Fb$ relies on the properties of the natural filtration of the canonical space and cannot be directly extended to a filtration $\Gb$ progressively enlarged by a random jump. In order to solve this problem, we construct the filtration $\Gb$ according to the canonical way in Section 6.5 of \cite{Bie-Rut}. Properties of this canonical construction allow the construction of a $\Gb$-sublinear conditional expectation, which is consistent with the one in \cite{Nut-Han} if restricted to $\Fb$. However, there are several additional technical difficulties in comparison to the construction on the canonical space. In particular, in order to be well-defined, the $\Gb$-sublinear conditional expectation requires integrability conditions, which are not necessary for the pathwise construction in \cite{Nut-Han}. This also implies that this extended sublinear operator only satisfies a weak version of dynamic programming principle or tower property in the general case, as in \cite{Peng}, as well as that it does not preserve integrability. However, the classic tower property and integrability invariance are shown to be satisfied in all cases of most common credit and insurance contracts. In Appendix \ref{app: other condition tower property}, we discuss further sufficient conditions, which guarantee the classic tower property. We refer to Section \ref{sec: construction G-conditional expectation}, Appendices \ref{app: counterexample} and \ref{app: other condition tower property} for a thorough discussion on these issues.
	
	Furthermore, we analyze for the first time the superhedging problem for a generic payment stream under model uncertainty and in continuous time. Superhedging dualities within the context of nondominated probability family have been intensively studied in recent years, e.g. \cite{Nut-Son}, \cite{Neu-Nut-sup}, \cite{Pos-Roy}, \cite{Dol-son}, \cite{Bay-Hua}, \cite{Hou-Obl}, \cite{Nut-rob} and \cite{Hob-Neu}. However, duality results achieved in these papers are mostly limited to the initial time and can be applied only to contingent claims. The superhedging problem for a generic payment stream, which is typically the case of credit or insurance cash flows, is studied only without model uncertainty and mostly in discrete time, e.g. \cite{Foll-Sch}, \cite{Penn} and \cite{Penn-dual}. Here we investigate dynamic robust superhedging duality for a generic payment stream in continuous time with respect to a nondominated probability family. Still in a dynamic way, we define separately global and local superhedging strategies and prices, which we are able to determine as a consequence of our duality results. %i.e. we call robust global superhedging strategies for a payment stream the ones which superhedge the total payment over every random interval. We call robust local superhedging strategies over a given random interval the ones which superhedge the total payment  over every sub-random intervals. Robust global and local superhedging prices are defined correspondingly. 	As a consequence of the duality results, we are able to determine both global and local superhedging price, as well as to prove the existence of optimal superhedging strategies.
	These results are first shown in standard setting and then extended to the robust reduced-form framework. In view of the superhedging results, the constructed $\Gb$-sublinear conditional expectation can be considered as a pricing operator for insurance and credit risk products.
	We would like to emphasize that our definitions and results hold without changes also in the case without model uncertainty, i.e. when a specific prior is considered. 
	
	The paper is organized as follows. In Section \ref{sec: reduced-form setting}, we construct a consistent robust reduced-form framework based on the canonical construction in \cite{Bie-Rut}. As the main result, we define explicitly sublinear conditional expectation on the progressively enlarged filtration and analyze its properties. The constructed operator is then applied to the valuation of credit and insurance contracts. In Section \ref{sec: superhedging}, we formulate the robust superhedging problem for payment streams in continuous time. We determine the robust superhedging price and prove the existence of optimal robust superhedging strategies first in the standard setting on the canonical space with the natural filtration and then in the reduced-form framework. In Appendix \ref{app: counterexample} we provide a counterexample showing that the classic tower property does not hold in full generality, while in Appendix \ref{app: other condition tower property}, we state sufficient conditions beside the ones in Section \ref{sec: reduced-form setting}, which guarantee the validity of the tower property.

\section{Reduced-form framework under model uncertainty}\label{sec: reduced-form setting}

In this section, we introduce the reduced-form setting under model uncertainty. 
We note that, the standard framework under model uncertainty considers only the canonical space endowed with the natural filtration, and do not allow to treat more general filtrations, see \cite{Aks-Hou} for a discussion on this point. In \cite{Aks-Hou}, the case of initial enlargement is solved while the case of progressive enlargement of filtration remains open. This issue arises in credit and insurance market modeling, when we want to model an event which occurs as a surprise and is itself not observable under the reference information flow, represented by a filtration $\Fb$, but has an $\Fb$-adapted intensity process.
Here we propose a solution for this problem by using the canonical construction in Section 6.5 of \cite{Bie-Rut} to introduce a random time $\tilde \tau$, which is not an $\Fb$-stopping time but admits an $\Fb$-adapted intensity, and extend  the concept of sublinear conditional expectation on the filtration progressively enlarged by this random time.
We first recall the setting in \cite{Nut-Han}.

\subsection{$(\P, \Fb)$-conditional expectation}\label{sec: Fb cond exp}

Let $\Omega = D_0(\R_+, \R^d)$ be the space of càdlàg functions $\omega = (\omega_t)_{t \geqslant 0}$ in $\R^d$ which start from zero. Equipped with metric induced by the Skorokhod topology, $\Omega$ is a Polish space, i.e. a complete separable metrizable space. 
We denote by $\F := \BOmega$ the Borel $\sigma$-algebra and by $\P(\Omega)$ the set of all probability measures on $(\Omega, \F)$. On $\mathcal{P}(\Omega)$ we consider the topology of weak convergence. 
According to Prokhorov's theorem (see e.g. \cite{Pro}, \cite{Del-Mey} and \cite{Bil}), $\mathcal{P}(\Omega)$ inherits from $\Omega$ the property of being a Polish space with the Lévy–-Prokhorov metric.

We consider the canonical process $B := (B_t)_{t \geqslant 0}$, where $B_t(\omega) := \omega_t$, $t \geqslant 0$, and denote its raw filtration by $\Fb = (\F_t)_{t \geqslant 0}$. It is easy to see that $\F_0=\{\emptyset, \Omega\}$ and $\F_{\infty} := \bigvee_{t \geqslant 0} \F_t = \F$. 
For every $P \in \mathcal{P}(\Omega)$ and $t \in \overline{\R}_+$, we denote by $\mathcal{N}_t^P$ the collection of sets which are $(P, \F_t)$-null and define 
\[
	\F_t^* := \F_t \vee \mathcal{N}_t^{*}, \ \ \ \  \mathcal{N}_t^{*} := \bigcap_{P \in \P(\Omega)} \mathcal{N}_t^P.
\]
The corresponding universally completed filtration is denoted by $\Fb^* := (\F^*_t)_{t \geqslant 0}$. Furthermore, for every $P \in \P(\Omega)$ the usual $P$-augmentation is denoted by $\Fb^P_+$, i.e. $\Fb^P_+$ is the right continuous version of $\Fb^P:= (\F_t^{P})_{t \geqslant 0}$, with
\[
	\F^P_t := \F_t \vee \mathcal{N}_\infty^P,  \ \ \ \ t \geqslant 0.
\]
Trivially, the above enlargements of the raw filtration are ordered in the following way
\begin{equation}\label{eq: inclusion filtrations}
	\F_t \subseteq \F_t^* \subseteq  \F_t^P \subseteq \F_{t+}^P, \ \ \ \ t \geqslant 0, \ \ \ \ P \in \P.
\end{equation}
Let $\P \subseteq \P(\Omega)$ be a generic nonempty set, we define the following $\sigma$-algebra
\[
	\F^{\P} := \F \vee \mathcal{N}_\infty^{\P}, \ \ \ \  \mathcal{N}^{\P}_\infty := \bigcap_{P \in \P}\mathcal{N}_\infty^P,
\] 
We denote by $L^0(\Omega)$ the space of all real-valued $\F^{\P}$-measurable functions and define the upper expectation $\mathcal{E} :  L^0(\Omega) \rightarrow \overline{\R}$ associated to $\P$ by
\begin{equation}\label{eq: representation upper exp}
	\mathcal{E} (X) := \sup_{P \in \mathcal{P}} \E^{P} [X], \ \ \ \  X \in L^0(\Omega),
\end{equation}
where for every $P \in \P$, we set $\E^{P} [X] := \E^{P} [X^+] - \E^{P} [X^-]$ if $\E^{P} [X^+]$ or $\E^{P} [X^-]$ is finite, and we use the convention $\E^{P} [X] := - \infty$ if $\E^{P} [X^+] = \E^{P} [X^-] = + \infty$, as in \cite{Son-Tou-qua}.

\begin{rem}
	Throughout the paper, all results also hold if the space $D_0(\R_+, \R^d)$ is replaced by $C_0(\R_+, \R^d)$, i.e. the space of continuous functions $\omega = (\omega_t)_{t \geqslant 0}$ in $\R^d$ which start from zero, equipped with the topology of locally uniform convergence. Since there is no ambiguity, we keep the notations $B$ and $\Fb$ for the canonical process on $C_0(\R_+, \R^d)$ and its natural filtration, respectively. %Throughout the paper if not otherwise specified, all relations hold in a pathwise sense.
\end{rem}

We now recall the pathwise construction in \cite{Nut-Han} of conditional expectation on $\Omega$ with respect to the filtration $\Fb$ and a family $\P \subseteq \P(\Omega)$ of probability measures. For notation simplicity, we consider only the case when the parametrized families in Assumption 2.1 of \cite{Nut-Han} have no dependence on the parameters. %We emphasize that all the results in the rest of this paper hold also in the general setting of \cite{Nut-Han}. 
As noted in e.g. \cite{Bia-Bou}, \cite{Neu-Nut-non} and \cite{Nut-rob}, the results in \cite{Nut-Han} hold both on the space $D_0(\R_+, \R^d)$ and on the space $C_0(\R_+, \R^d)$. %We note that here it is crucial to use our specification of $\Omega$ and $\Fb$. 

We introduce the following notations according to \cite{Nut-Han}. Let $\tau$ be a finite-valued $\Fb$-stopping time and $\omega \in \Omega$. For every ${\omega}' \in \Omega$, the concatenation $\omega \otimes_\tau {\omega}' := ((\omega \otimes_\tau {\omega}')_t)_{t \geqslant 0}$ of $(\omega, {\omega}')$ at $\tau$ is given by
\begin{equation}\label{eq: pasting}
	\left( \omega \otimes_\tau {\omega}' \right)_t := \omega_t \mathbf{1}_{[0,\tau(\omega))}(t) + \left( \omega_{\tau(\omega)} + {\omega}'_{t - \tau(\omega)} \right) \mathbf{1}_{[\tau(\omega), + \infty)}(t), \ \ \ \  t \geqslant 0.
\end{equation}
For every function $X$ on $\Omega$, we define the following function
\begin{equation}\label{eq: def concatenation}
	X^{\tau, \omega} ({\omega}') := X (\omega \otimes_\tau {\omega}'), \ \ \ \  {\omega}' \in \Omega.
\end{equation}
Similarly, for every probability measure $P$ we set
\[
	P^{\tau, \omega} (A) := P^\omega_\tau(\omega \otimes_\tau A), \ \ \ \  A \in \BOmega,
\]
which is still a probability measure, where $\omega \otimes_\tau A := \{ \omega \otimes_\tau {\omega}' : {\omega}' \in A \}$ and $P^\omega_\tau$ is the $\F_\tau$-conditional probability measure chosen to be	
\[
	P^\omega_\tau \left(\omega' \in \Omega : \omega' = \omega \text{ on } [0, \tau(\omega)] \right) = 1.
\]
Furthermore, we recall that a set of a Polish space is called analytic if it is the image of a Borel set of an other Polish space under a Borel-measurable mapping. A $\overline{\R}$-valued function $f$ on a Polish space is called upper semianalytic if $\{ f > c \}$ is analytic for every $c \in \R$. In particular, we note that all Borel sets are analytic and all Borel-measurable functions are upper semianalytic. 

\begin{assump}\label{assump: conditions on P (only one) for conditional sub exp}
	For every finite-valued $\Fb$-stopping time $\tau$, the family $\P$ satisfies the following conditions:
	\begin{enumerate}
		\item measurability: the set $\P \in \P(\Omega)$ is analytic;
		
		\item invariance: $P^{\tau, {\omega}} \in \P $ for $P$-a.e. ${\omega} \in \Omega$;
		
		\item stability under pasting: for every $\F_{\tau}$-measurable kernel $\kappa : \Omega \rightarrow \P(\Omega)$ such that $\kappa({\omega}) \in \P$ for $P$-a.e. ${{\omega}} \in \Omega$, the following measure
		\[
			\overline{P}(A) := \iint (\mathbf{1}_A)^{\tau, {\omega}} (\omega') \kappa(\ud \omega'; {\omega}) P(\ud {\omega}), \ \ \ \  A \in \BOmega,
		\]
		still belongs to $\P$.
	\end{enumerate}
\end{assump}

\begin{rem}
	As shown in \cite{Neu-Nut-non}, Assumption \ref{assump: conditions on P (only one) for conditional sub exp} is satisfied when the family $\P$ is generated by all semimartingale laws with differential characteristics taking values in a Borel-measurable set $\theta \subseteq \R^d \times \mathbb{S}^d_+ \times \mathcal{L}$, where $\mathbb{S}^d_+$ is the set of symmetric nonnegative definite $(d \times d)$-matrices and $\mathcal{L}$ is the set of all Lévy measures. In particular, this case includes the $G$-expectations introduced in \cite{Peng-gexp}. 
\end{rem}

The following proposition is a special case of Theorem 2.3 of \cite{Nut-Han}, when we restrict our attention to one family $\P$ satisfying Assumption \ref{assump: conditions on P (only one) for conditional sub exp}.

\begin{prop}\label{prop: construction F-cond sub exp (only one P)}
	Under Assumption \ref{assump: conditions on P (only one) for conditional sub exp},  for all finite-valued $\Fb$-stopping times $\sigma, \tau$ such that $\sigma \leqslant \tau$ and for every upper semianalytic function $X$ on $\Omega$, the function $\mathcal{E}_\tau (X)$ defined by
	\begin{equation}\label{eq: def cond expect}
		\mathcal{E}_\tau(X)(\omega) := \mathcal{E} (X^{\tau, \omega} ) = \sup_{P \in \mathcal{P}} \E^{P} [X^{\tau, \omega}], \ \ \ \  \omega \in \Omega
	\end{equation}
	is $\F^*_\tau$-measurable, upper semianalytic and satisfies the following consistency condition 
	\begin{equation}\label{eq: consist cond for t}
		\mathcal{E}_\tau(X) = \underset{P' \in \mathcal{P}(\tau; P)}{\text{ess sup}^P} \E^{P'}[X | \mathcal{F}_\tau]  \ \ \ \ P\text{-a.s. for all } P \in \P,
	\end{equation}
	where $\mathcal{P}(\tau; P) := \left\{ P' \in \P : P' = P \text{ on } \F_\tau \right \}$. Furthermore, the tower property holds, i.e. 
	\begin{equation}\label{eq: tower property F}
		\mathcal{E}_\sigma (X)(\omega) = \mathcal{E}_\sigma (\mathcal{E}_\tau (X))(\omega) \ \ \ \ \text{for all } \omega \in \Omega. 
	\end{equation}
\end{prop}

\begin{defn}
	We call the family of sublinear conditional expectations $(\mathcal{E}_t)_{t \geqslant 0}$ \emph{$(\P, \Fb)$-conditional expectation}.
\end{defn}

In the special case of $G$-setting introduced in \cite{Peng-gexp}, $G$-martingales are càdlàg, see e.g. \cite{Song}. However, under generic assumptions, the process $(\mathcal{E}_t(X))_{t \geqslant 0}$ with $X$ upper semianalytic is not always càdlàg. 
In the following proposition, we show an independent result which gives sufficient conditions for having $(\mathcal{E}_t(X))_{t \geqslant 0}$ càdlàg.  We recall that by Prokhorov's theorem, a family of probability measures is tight if and only if its weak closure is compact. In particular the probability measure family which generates the $G$-expectation is tight, see e.g. Proposition 49 in \cite{Den-Hu}.

\begin{prop}\label{prop: conditional expectation cadlag}
	If $\P$ is a tight family satisfying Assumption \ref{assump: conditions on P (only one) for conditional sub exp} and $X$ is an upper semianalytic function on $\Omega$ which is bounded and continuous $P$-a.s. for all $P \in \P$, then the process $(\mathcal{E}_t(X))_{t \geqslant 0}$ is càdlàg.
\end{prop}

\begin{proof}
	Let $A \in \mathcal{B}(\Omega)$  be a set such that $X$ is bounded and continuous on $A$ and $P(A^c)=0$ for every $P \in \P$.
	We start with the right continuity. Let $t \geqslant 0$ and $(t_n)_{n \in \N}$ be a sequence in $\R$ such that $t_n \downarrow t$. We want to show that for all $\omega \in \Omega$,
	\[
		\mathcal{E}_{t}(X)(\omega) = \lim_{n \rightarrow \infty} \mathcal{E}_{t_n}(X)(\omega).
	\]
	Consider $\omega \in \Omega$. By definitions (\ref{eq: def cond expect}) and (\ref{eq: def concatenation}) we have 
	\[
		\mathcal{E}_{t}(X)(\omega) = \mathcal{E}(X^{t, \omega}) = \sup_{P \in \P} E^P[X^{t, \omega}] = \sup_{P \in \P} \int X(\omega \otimes_t \omega') P (\ud \omega').
	\]
	For fixed $t$ and $\omega$, we define the concatenation function $c^{t, \omega} : \Omega \rightarrow \Omega$ by $c^{t, \omega}(\omega') := \omega \otimes_t \omega'$, $\omega' \in \Omega$. This function is uniformly continuous in $\omega'$ with respect to Skorokhod topology on $\Omega = D_0(\R_+, R^d)$\footnote{Or locally uniform convergence on $\Omega = C_0(\R_+, R^d)$.}. Namely, if we denote by $d$ the distance induced by Skorokhod topology on $\Omega$, we have that for every $\varepsilon > 0$, there is a $\delta > 0$ such that for all $\omega', \omega'' \in \Omega$ with $d(\omega', \omega'') < \delta$, it holds 
	\[
		d (\omega \otimes_t \omega', \omega \otimes_t \omega'') < \varepsilon.
	\]
	Indeed, it is sufficient to take $\delta = \varepsilon$. We note that $\delta = \varepsilon$ does not depend on the choice of $t$, hence in particular the sequence of functions $(c^{t_n, \omega})_{n \in \N}$ is equicontinuous.
	Furthermore, the sequence $(c^{t_n, \omega})_{n \in \N}$ converges to $c^{t, \omega}$ pointwisely, 
	\[
		d\left(\omega \otimes_{t_n} \omega', \omega \otimes_t \omega'\right) \xrightarrow{n \rightarrow \infty} 0 \ \ \ \ \text{for all } \omega' \in \Omega,
	\]	
	since $D_0(\R_+, R^d)$ is the space of càdlàg paths.
	Hence, by Ascoli-Arzelà Theorem, the sequence $(c^{t_n, \omega})_{n \in \N}$ converges to $c^{t, \omega}$ uniformly on every compact set $K \subseteq \Omega$, i.e. we have
	\[
		\sup_{\omega' \in K} d\left(c^{t_n, \omega}(\omega'), c^{t, \omega}(\omega') \right) = \sup_{\omega' \in K} d\left(\omega \otimes_{t_n} \omega', \omega \otimes_t \omega' \right) \xrightarrow{n \rightarrow \infty} 0.
	\]
	In particular, given a compact set $K \in \mathcal{B}(\Omega)$,  the composition $X^{t, \omega} = X \circ c^{t, \omega}$ is bounded and continuous on $A \cap K$, and $X^{t, \omega}$ is the uniform limit of $(X^{t_n, \omega})_{n \in \N}$, i.e. for every $\varepsilon > 0$, there is $N \in \N$ such that for all $n \geqslant N$, 
	\[
		|X(\omega \otimes_{t_n} \omega') - X(\omega \otimes_t \omega')| < \varepsilon \ \ \ \ \text{for every } \omega' \in A \cap K.
	\]
		
	As a consequence, on one hand, for every $n \in \N$, the function $f^n$ defined by $f^n(P) := E^P[X^{t_n, \omega}]$, $P \in \P(\Omega),$ is continuous in $P$ with respect to Lévy–-Prokhorov metric on $\P(\Omega)$, since it coincides with the metric induced by weak convergence of measures. Hence the restriction $f^n|_{\P}$ is still continuous.
	On the other hand, the tightness of $\P$ yields that there is a compact set $K \in \mathcal{B}(\Omega)$ such that $P(K^c) < \frac{\varepsilon}{4C}$ for all $P \in \P$, where $C$ is such that $|X(\omega)| \leqslant C$ for every $\omega \in A$. For $n$ big enough, since  $X^{t, \omega}$ is the $P$-a.s. uniform limit of $(X^{t_n, \omega})_{n \in \N}$ on $A \cap K$, we have
	\begin{align*}
		|E^P[X^{{t_n}, \omega} ] - E^P[X^{{t}, \omega} ]| &\leqslant E^P[|X^{{t_n}, \omega} - X^{{t}, \omega}| ] \\
		&= E^P[\mathbf{1}_{A \cap K} |X^{{t_n}, \omega} - X^{{t}, \omega}| ] + E^P[\mathbf{1}_{A \backslash K } |X^{{t_n}, \omega} - X^{{t}, \omega}| ]\\ &< \frac{\varepsilon}{2} +  \frac{\varepsilon}{4C} \cdot 2C = \varepsilon \ \ \ \ \text{for all } P \in \P.
	\end{align*}
	Hence for all $\omega \in \Omega$,
	\begin{align*}
		\mathcal{E}(X^{{t}, \omega} ) =& \sup_{P \in \P} E^P [\lim_{n \rightarrow \infty}X^{{t_n}, \omega}] = \sup_{P \in \P}\lim_{n \rightarrow \infty} E^P [X^{{t_n}, \omega}]\\ =& \lim_{n \rightarrow \infty}\sup_{P \in \P} E^P [X^{{t_n}, \omega}] = \lim_{n \rightarrow \infty} \mathcal{E}(X^{{t_n}, \omega}).
	\end{align*}
	A similar argument also shows the existence and finiteness of the left limit, which concludes the proof.
\end{proof}

\begin{rem} 
Proposition 4.5 of \cite{Nut-Son} introduces a family of sublinear operators depending on a filtration different from $\Fb^*$, i.e. 
\[
	\left(\Fb_{t+} \cup \mathcal{N}_T^\mathcal{P} \right)_{t\in[0,T]},
\]
where $\mathcal{N}_T^{\P}$ is the collection of sets which are $(P, \F_T)$-null for all $P \in \P$. 
In this way the resulting sublinear operator is càdlàg in $t$. However, for the applications which we consider in this paper, it is fundamental to work with the filtration $\Fb^*$, since it represents the information available to the agents.
\end{rem}

\subsection{Space construction}\label{sec: space construction}

We keep the same notations in Section \ref{sec: Fb cond exp}. In this section we follow the canonical space construction in Section 6.5 of \cite{Bie-Rut} to introduce a random time $\tilde \tau$, which is not an $\Fb$-stopping time but has an $\Fb$-progressively measurable intensity process $\mu$, to represent a totally unexpected default or decease time under model uncertainty.
Let $\hat{\Omega}$ be an additional Polish space equipped with its Borel $\sigma$-algebra $\mathcal{B}(\hat \Omega)$. We now consider the product measurable space $(\tilde{\Omega}, \G) := (\Omega \times \hat{\Omega}, \mathcal{B}(\Omega) \otimes \mathcal{B}(\hat \Omega) )$, and use the notation $\tilde \omega = (\omega, \hat{\omega})$ for $\omega \in \Omega$ and $\hat{\omega} \in \hat{\Omega}$.
The following standard conventions are made on the product space $(\tilde{\Omega}, \G )$.  For every function or process $X$ on $(\Omega, \mathcal{B}(\Omega))$, we consider its natural immersion into the product space, i.e. $X(\tilde{\omega}) := X(\omega)$ for all $\omega \in \Omega$, similarly for $(\hat{\Omega}, \mathcal{B}(\hat \Omega))$. For every sub-$\sigma$-algebra $\mathcal{A}$ of $\mathcal{B}(\Omega)$, we consider its natural extension $\mathcal{A} \otimes \{ \emptyset, \hat{\Omega} \}$ as a sub-$\sigma$-algebra of $\G$ on the product space, similarly for sub-$\sigma$-algebras of $\mathcal{B}(\hat \Omega)$. To avoid cumbersome notations, when there is no ambiguity, $\mathcal{A} \otimes \{ \emptyset, \hat{\Omega} \}$ is still denoted by $\mathcal{A}$. 

On $(\hat{\Omega}, \mathcal{B}(\hat \Omega))$ we fix a probability measure $\hat{P}$ such that $(\hat{\Omega}, \mathcal{B}(\hat \Omega) ), \hat{P})$ is an atomless probability space, i.e. there exists a random variable with an absolutely continuous distribution, and let $\xi$ be a Borel-measurable surjective random variable
\[
	\xi : (\hat{\Omega}, \mathcal{B}(\hat \Omega), \hat P) \rightarrow ([0,1], \mathcal{B}([0,1])),
\]
with uniform distribution, that is
\[
	\xi \sim U([0,1]).
\]
Without loss of generality we assume $\mathcal{B}(\hat \Omega) = \sigma(\xi)$. 

\begin{rem}
	We note that the space $(\hat{\Omega}, \mathcal{B}(\hat \Omega), \hat{P})$ can be set canonically as 
	\[
		\left([0,1], \mathcal{B}([0,1]), U([0,1])\right),
	\]
	with $\xi$ the identity function on $[0,1]$.
\end{rem}

\noindent We denote by $\mathcal{P}(\tilde{\Omega})$ the set of all probability measures on $(\tilde{\Omega}, \G)$ and consider the following family of probability measures 
\begin{equation}\label{eq: definition tilde P}
	\tilde{\mathcal{P}} := \left\{ \tilde P \in \P (\tilde{\Omega}) : \tilde P = P \otimes \hat{P}, \ P \in \mathcal{P} \right\}.
\end{equation}

 On $(\Omega, \mathcal{B}(\Omega))$ let $\Gamma := (\Gamma_t)_{t \geqslant 0}$ be a real-valued, $\Fb$-adapted, continuous and increasing process such that  $\Gamma_0 = 0$ and $\Gamma_{\infty} = + \infty$. In particular, $\Gamma$ can be represented by
\[
	\Gamma_t := \int_0^t \mu_s \ud s, \ \ \ \  t \geqslant 0,
\]
where $\mu := (\mu_t)_{t \geqslant 0}$ is a nonnegative $\Fb$-progressively measurable process such that for all $t \geqslant 0$ and for all $\omega \in \Omega$,
\[
	\int_0^t |\mu_s|(\omega) \ud s < \infty.
\]
We define
\begin{align*}
	\tilde \tau :&= \inf \{ t \geqslant 0 : e^{- \Gamma_t} \leqslant \xi \} = \inf \{ t \geqslant 0 : \Gamma_t \leqslant - \ln \xi \}
\end{align*}
on $\tilde{\Omega} = \Omega \times \hat{\Omega}$, with the convention $\inf \emptyset = \infty$.

\begin{rem}\label{rem: tau surjective}
	An immediate consequence of the above assumptions is that $\tilde \tau (\omega, \cdot)$ is a surjective function on $\R_+$ for every fixed $\omega \in \Omega$.
\end{rem}

\begin{lemma}\label{lemma: equivalence set for tau}
	For every $t \geqslant 0$, we have $\{ \tilde \tau \leqslant t \} = \{ e^{- \Gamma_t} \leqslant \xi \}$.
\end{lemma}

\begin{proof}
	We note that $\{ e^{- \Gamma_t} \leqslant \xi \} \subseteq \{ 	\tilde \tau \leqslant t \}$ always holds. The other inclusion follows from 
	\[
		\tilde \tau = \min \{ s \geqslant 0 : e^{- \Gamma_s} \leqslant \xi \},
	\]
	since $\Gamma$ is continuous.
\end{proof}

Under every $\tilde{P} \in \tilde{\P}$, we define the $\tilde{P}$-hazard process $\Gamma^{\tilde P} := (\Gamma^{\tilde P}_t)_{t \geqslant 0}$ by
\[
	\Gamma^{\tilde P}_t := - \ln  \tprobcondf{\tilde \tau > t}, \ \ \ \ \ t \geqslant 0.
\]
The following proposition is a natural but important consequence of the above construction.

\begin{prop}\label{prop: gamma aggregator}
	The process $\Gamma$ is a $\tilde P$-a.s. version of $\tilde{P}$-hazard process $\Gamma^{\tilde P}$ for every $\tilde{P} \in \tilde{\P}$.
\end{prop}

\begin{proof}
	By Lemma \ref{lemma: equivalence set for tau},
	\[
		\{\tilde \tau > t \} = \{ e^{- \Gamma_t} > \xi \} \ \ \ \ \text{for all } t \geqslant 0.
	\]
	Hence for every $t \geqslant 0$ and for every $\tilde{P} \in \tilde{\P}$ with $\tilde P = P \otimes \hat{P}$, it holds
	\begin{align*}
		e^{- \Gamma^{\tilde P}_t(\omega)} &= \tprobcondf{\tilde \tau > t}(\omega) = \tprobcondf{e^{- \Gamma_t} > \xi}(\omega) \\
				&\stackrel{(i)}{=} \left. \tilde{P} \left( e^{-x} > \xi \right) \right|_{x = \Gamma_t(\omega)} = \left. {\hat{P}} \left( e^{-x} > \xi \right) \right|_{x = \Gamma_t(\omega)}\\
				&\stackrel{(ii)}{=} \left. e^{-x} \right|_{x = \Gamma_t(\omega)} \\
				&= e^{- \Gamma_t(\omega)} \ \ \ \ \text{for } \tilde{P}\text{-a.e. } \omega,
	\end{align*}
	where equality (i) follows from the independence between $\xi$ and $\F_t$ under each $\tilde P \in \tilde{\P}$, and equality (ii) follows from the fact that $\xi$ has uniform distribution on $(\hat{\Omega}, \hat{\F}, \hat{P})$. The continuity of $\Gamma$ yields
	\[
		\Gamma^{\tilde P} = \Gamma \ \ \ \ \tilde P\text{-a.s. for all } \tilde P \in \tilde \P,
	\]
	which concludes the proof.
\end{proof}

On the product space $\tilde{\Omega}$, we consider the filtration $\Hb := (\mathcal{H}_t)_{t \geqslant 0}$ generated by the process $H := (H_t)_{t \geqslant 0}$ defined by
\[
	H_t := \mathbf{1}_{\{ \tilde \tau \leqslant t \}}, \ \ \ \  t \geqslant 0,
\]
and the enlarged filtration $\Gb := (\mathcal{G}_t)_{t \geqslant 0}$ defined by $\G_t := \F_t \vee \H_t$, $t \geqslant 0$. In particular, we have $\G = \F_\infty \otimes \sigma(\xi) = \H_\infty \vee \F_\infty = \sigma(\tilde \tau) \vee \F_\infty$.
By construction $\tilde \tau$ is an $\Hb$-stopping time as well as a $\Gb$-stopping time, but not an $\Fb$-stopping time. The filtration $\Fb$ can be interpreted as the reference information flow, while the filtration $\Gb$ represents the minimal information flow of the extended market including default information.   As in Section \ref{sec: Fb cond exp}, for every $\tilde P \in P(\tilde{\Omega})$ we denote by $\Gb^{*}$, $\Gb^{\tilde P}$ and $\Gb^{\tilde P}_+$ the corresponding enlargements of the raw filtration $\Gb$. Similarly to (\ref{eq: inclusion filtrations}), we have 
\[
	\G_t \subseteq \G_t^* \subseteq \G_t^{\tilde P} \subseteq \G_{t+}^{\tilde P}, \ \ \ \ t \geqslant 0, \ \ \ \ \tilde P \in \tilde{\P}.
\]

\subsection{$(\tilde{\P}, \Gb)$-conditional expectation}\label{sec: construction G-conditional expectation}

In this section, we give a construction of sublinear conditional expectations with respect to the filtration $\Gb$ and the family of probability $\tilde{\P}$  introduced in (\ref{eq: definition tilde P}). These will be denoted by $(\mathcal{\tilde E}_{t})_{t \geqslant 0}$ and called $(\tilde{\P}, \Gb)$-conditional expectation.
Such construction is motivated by the results in Section \ref{sec: Fb cond exp} and should reflect the underlying structure of the space construction in Section \ref{sec: space construction}. %Moreover, this also represents a significant extension of the construction of sublinear conditional expectation as proposed by \cite{Nut-Han} to a completely different setting. 
According to e.g. \cite{Son-Tou-qua}, \cite{Coh}, \cite{Son-Tou-well}, \cite{Son-Tou-dual} and \cite{Nut-Han}, the family $(\mathcal{\tilde E}_{t})_{t \geqslant 0}$ should satisfy the following necessary consistency condition: for every $t \geqslant 0$ and $\G$-measurable function $\tilde X$ on $\tilde \Omega$,
\begin{equation}\label{eq: consistency condition G-cond exp}
	\mathcal{\tilde E}_{t}(\tilde X) = \underset{\tilde P' \in \mathcal{\tilde P}({t}; \tilde P)}{\text{ess sup}^{\tilde P}} \E^{\tilde{P'}}[\tilde X | \mathcal{G}_{t}] \ \ \ \ \tilde P\text{-a.s. for all } \tilde P \in \tilde{\P},
\end{equation}
where $\mathcal{\tilde P}({t}; \tilde P) := \left\{ \tilde{P'} \in \tilde{\P} : \tilde{P'} = \tilde{P} \text{ on } \G_{t} \right \}$.
We emphasize that this cannot be done by using exactly the same method proposed in \cite{Nut-Han} and summarised in Section \ref{sec: Fb cond exp}, even if we choose $\hat{\Omega} = D_0(\R+, \R^d)$ or $\hat{\Omega} = C_0(\R+, \R^d)$. Indeed, the approach in \cite{Nut-Han} is based on some special properties of the natural filtration generated by the canonical process, e.g. Galmarino's test, which the filtration $\Gb$ does not have.
Nevertheless, we are able to extend the results of \cite{Nut-Han}  to the setting of Section \ref{sec: space construction}, and construct a consistent $(\tilde{\P}, \Gb)$-conditional expectation.
As in \cite{Peng}, we show that the family $(\mathcal{\tilde E}_{t}(\tilde X))_{t \geqslant 0}$ satisfies a weak form of time-consistency, called also dynamic programming principle or tower property, i.e.
\begin{equation}\label{eq: weak tower property}
	\mathcal{\tilde E}_{s}(\mathcal{\tilde E}_{t}(\tilde X)) \geqslant \mathcal{\tilde E}_{s}(\tilde X) \ \ \ \ \text{for all } 0 \leqslant s \leqslant t \ \tilde P \text{-a.s. for all } \tilde P \in \tilde{\P}.
\end{equation}
From an economical point of view, by using $(\tilde{\mathcal{E}}_t)_{t \geqslant 0}$ as pricing functional, the weak tower property (\ref{eq: weak tower property}) can be interpreted as: making valuation of an evaluated future price is more conservative than making direct valuation of the price.
We provide some sufficient conditions such that the classic tower property holds. These include all cases of often used credit and insurance contracts, as explained in Section \ref{sec: main products}.\\
%For the sake of convenience, from this section onwards we restrict our attention to deterministic times.  \\
As in Section \ref{sec: Fb cond exp}, we use the corresponding notations and denote the upper expectation associated to $\tilde \P$ by $\tilde{\mathcal{E}}$, i.e. 
\begin{equation}\label{eq: def tilde E}
		\tilde{\mathcal{E}}(\tilde X) := \sup_{\tilde P \in \tilde{\P}} E^{\tilde P} [\tilde X], \ \ \ \ \tilde X \in L^0(\tilde \Omega).
\end{equation}
Let $\G^{P} := \G \vee \mathcal{N}^{P}_\infty$, $P \in \P$, and $\G^{\P} := \G \vee \mathcal{N}^{\P}_\infty$. We introduce the following sets
\begin{align*}
	L^1_{\tilde P}(\tilde \Omega) :=& \{\tilde X \ | \ \tilde X :(\tilde \Omega, \G^{P}) \rightarrow (\R, \mathcal{B}(\R)) \text{ measurable function such that }\\
	& E^{\tilde P}[|\tilde X|] < \infty\},
\end{align*}
for every $\tilde P \in \tilde{\P}$, and
\begin{align*}
	L^1(\tilde \Omega) := \{&\tilde X \ | \ \tilde X :(\tilde \Omega, \G^{{\P}}) \rightarrow (\R, \mathcal{B}(\R)) \text{ measurable function such that }\\
	&\left.\tilde{\mathcal{E}}(|\tilde X|) < \infty \right\}.
\end{align*}
We emphasize that in the above definitions we only consider $(\Omega, \G^{{P}})$-measurable (or $(\Omega, \G^{{\P}})$-measurable resp.) functions, and not $(\Omega, \G^{{\tilde P}})$-measurable (or $(\Omega, \G^{\tilde {\P}})$-measurable resp.) functions, see also Remark \ref{rem: replacement F}.\\
Given $t \geqslant 0$, every real-valued function $\tilde X$ on $\tilde \Omega$ can be decomposed in 
\begin{equation*} 
	\tilde X = \Ind{\tilde \tau \leqslant t} \tilde X + \Ind{\tilde \tau > t}\tilde X.
\end{equation*}
Corollary 5.1.2 of \cite{Bie-Rut}, which holds without the usual conditions on the filtrations, together with Proposition \ref{prop: gamma aggregator} shows that if $\tilde X \in L^1(\tilde \Omega)$, then for every $\tilde P \in \tilde \P$,
\begin{equation}\label{eq: Bielecki decomposition}
	E^{\tilde P} [\tilde X | \G_t] = \Ind{\tilde \tau \leqslant t} E^{\tilde P} [ \tilde X | \sigma(\tilde \tau) \vee \F_t] + \Ind{\tilde \tau > t}  e^{\Gamma_t}  E^{\tilde P}[ \Ind{\tilde \tau > t}  \tilde X | \F_t] \ \ \ \ \tilde P \text{-a.s.}
\end{equation}
Our goal is to find a representation of (\ref{eq: Bielecki decomposition}) with the right-hand side reduced to conditional expectations restricted to $\Omega$. This will play a fundamental role in the definition of conditional expectation on $\tilde \Omega$. The following Lemma solves the problem for the second term on the right-hand side of (\ref{eq: Bielecki decomposition}). For the sake of simplicity we use a slight abuse of notation and denote 
\begin{equation}\label{eq: notation E*}
	E^{\hat{P}}[ \tilde X] (\omega) := \int_{\hat{\Omega}} \tilde X (\omega, \hat{\omega}) \hat{P} ( \ud \hat{\omega}), \ \ \ \ \omega \in \Omega.
\end{equation}

\begin{lemma}\label{lemma: representation second addend}
	Let $t \geqslant 0$ and $\tilde P = P \otimes \hat{P}$. If $\tilde X \in L^1_{\tilde P}(\tilde \Omega)$, then
	\begin{equation*}
		 E^{\tilde P}[ \tilde X | \F_t] = E^{P}[ E^{\hat{P}}[ \tilde X] | \F_t] \ \ \ \ \tilde P\text{-a.s.}
	\end{equation*}
\end{lemma}

\begin{proof}
	It is sufficient to see that for any $A \in \F_t$, by the Fubini-Tonelli theorem we have
	\begin{align*}
		\int_{A \times \hat{\Omega}} \tilde X (\omega, \hat{\omega})\tilde P(\ud (\omega, \hat{\omega})) =& \int_{A} \int_{\hat{\Omega}} \tilde X (\omega, \hat{\omega})  \hat{P}(\ud \hat{\omega})  P(\ud \omega)\\ 
		=& \int_{A} E^{{\hat{P}}} [ \tilde X ](\omega)  P(\ud \omega)\\
		=& \int_{A \times \hat{\Omega}} E^{P}[ E^{\hat{P}}[ \tilde X]| \F_t](\omega)   \tilde P(\ud(\omega, \hat{\omega})),
	\end{align*}
	where we use the notation introduced in (\ref{eq: notation E*}).
\end{proof}

Now we focus on the first term on the right-hand side of (\ref{eq: Bielecki decomposition}).

\begin{lemma}\label{lemma: representation first addend}
	Let $t \in \overline{\R}_+$. If $\tilde X$ is a real-valued $\sigma(\tilde \tau) \vee \F_t$-measurable function on $\tilde \Omega$, then there exists a unique measurable function 
	\[
		\varphi: (\R_+ \times \Omega \ , \ \mathcal{B}(\R_+) \otimes \F_t) \rightarrow (\R, \mathcal{B}(\R)),
	\]
	such that 
	\begin{equation}\label{eq: representation X with phi}
		\tilde X(\omega, \hat{\omega}) = \varphi (\tilde \tau(\omega, \hat{\omega}), \omega), \ \ \ \  (\omega, \hat{\omega}) \in \tilde{\Omega}.
	\end{equation}
\end{lemma}

\begin{proof}
	The uniqueness of $\varphi$ which satisfies (\ref{eq: representation X with phi}) follows directly from the surjectivity of $\tilde \tau$ for every fixed $\omega \in \Omega$, see Remark \ref{rem: tau surjective}. Indeed, if $\varphi$ and $\psi$ are two functions such that 
	\[
		\varphi (\tilde \tau(\omega, \hat{\omega}), \omega) = \psi (\tilde \tau(\omega, \hat{\omega}), \omega) \ \ \ \ \text{for all } (\omega, \hat{\omega}) \in \tilde{\Omega},
	\]
	then for every $(x, \omega) \in \R_+ \times \Omega$, the surjectivity of $\tilde \tau$ for every fixed $\omega \in \Omega$ yields that there is an $\hat \omega \in \hat{\Omega}$ such that $\tau(\omega, \hat{\omega})=x$. Consequently 
	\[
		\varphi (x, \omega) = \psi (x, \omega) \ \ \ \ \text{for all }(x, \omega) \in \R_+ \times \Omega.
	\]
	
	\noindent Now we consider the following set
	\begin{align*}
		E =& \{ \tilde X \ | \ (\tilde \Omega, \sigma(\tilde \tau) \vee \F_t) \rightarrow (\R, \mathcal{B}(\R)), \ \tilde X \text{ of the form (\ref{eq: representation X with phi})} \},
	\end{align*}
	and show that it contains a monotone class.
	The set $E$ clearly contains all constants and is closed under linear operations. Furthermore, all indicator functions of a $\pi$-system which generates $\sigma(\tilde \tau) \vee \F_t$ belong to $E$. Now let $(\tilde{X}_n)_{n \in \N}$ be a sequence in $E$ such that $\tilde{X}_n(\tilde \omega) \uparrow \tilde X (\tilde \omega)$ for all $\tilde \omega \in \tilde{\Omega}$, where $\tilde X$ is a bounded function. For every $n \in \N$, we have $\tilde{X}_n(\omega, \hat{\omega}) = \varphi_n(\tilde \tau(\omega, \hat{\omega}), \omega)$ for all $(\omega, \hat{\omega}) \in \tilde{\Omega}$, where $\varphi_n$ is a measurable function 
	\[
		\varphi_n: (\R_+ \times \Omega \ , \ \mathcal{B}(\R_+) \otimes \F_t) \rightarrow (\R, \mathcal{B}(\R)).
	\]
	By Remark \ref{rem: tau surjective} and the boundedness of $\tilde X$, we note that the function
	\begin{equation}\label{eq: lim phi}
		\varphi (z,\omega) := \lim_{n \rightarrow \infty} \varphi_n (z, \omega), \ \ \ \ z \in \R_+,  \ \omega \in \Omega,
	\end{equation}
	is well defined and finite. In particular, $\varphi$ is also $(\mathcal{B}(\R_+) \otimes \F_t)$-measurable. By applying again Remark \ref{rem: tau surjective},  $\tilde X$ can be represented by
	\[
		\tilde X(\omega, \hat{\omega}) = \varphi(\tilde \tau(\omega, \hat{\omega}), \omega), \ \ \ \ (\omega, \hat{\omega}) \in \tilde{\Omega}.
	\]
	Hence $X$ belongs to $E$ as well. By the Monotone Class theorem, the set $E$ contains all bounded $\sigma(\tilde \tau) \vee \F_t$-measurable functions.
	
	Furthermore, every nonnegative $\sigma(\tilde \tau) \vee \F_t$-measurable function $\tilde X$ is the pointwise limit of a nondecreasing sequence of simple functions, i.e. there exists a sequence of simple functions $(\tilde{X}_n)_{n \in \N}$ such that $\tilde{X}_n(\tilde \omega) \uparrow \tilde X (\tilde \omega)$ for all $\tilde \omega \in \tilde{\Omega}$. In particular, by the argument above, if $\tilde{X}_n(\omega, \hat{\omega}) = \varphi_n(\tilde \tau(\omega, \hat{\omega}), \omega)$ for all $(\omega, \hat{\omega}) \in \tilde{\Omega}$, by defining $\varphi$ as the pointwise limit of $(\varphi_n)_{n \in \N}$ as in (\ref{eq: lim phi}), we conclude that all nonnegative $\sigma(\tilde \tau) \vee \F_t$-measurable functions have representation (\ref{eq: representation X with phi}). The results can be extended to all $\sigma(\tilde \tau) \vee \F_t$-measurable functions since $\tilde X = \tilde{X}^+ + \tilde{X}^-$.
\end{proof}

\begin{rem}\label{rem: replacement F}
	Lemma \ref{lemma: representation first addend} can be carried out without changes if $\tilde X$ is $\G^P$-measurable or $\G^{\P}$-measurable.  In such case $\varphi$ is $(\mathcal{B}(\R_+) \otimes \F^{P}_\infty)$-measurable or $(\mathcal{B}(\R_+) \otimes \F^{\P}_\infty)$-measurable, respectively. However, it does not hold if $\tilde X$ is $\G^{\tilde P}$-measurable with $\G^{\tilde P} := \G \vee \mathcal{N}^{\tilde P}_\infty$ or $\G^{\tilde{\P}}$-measurable with $\G^{\tilde{\P}} = \G \vee \mathcal{N}^{\tilde{\P}}_\infty$, respectively. The reason is analogue to the case of the classic Doob-–Dynkin lemma, which states that if $X, Y$ are two real-valued measurable functions and $Y$ is $\sigma(X)$-measurable, then there is a Borel-measurable function $f$ such that $Y = f(X)$. This representation does not hold pathwisely if $\sigma(X)$ is completed with null sets of some measure $Q$, i.e. if $\sigma(X)$ is replaced by $\sigma(X) \vee \mathcal{N}^Q$. Indeed, it is sufficient to take $Y = \mathbf{1}_A$ with $A \in \mathcal{N}^Q$ as a counterexample.
\end{rem}

\begin{lemma}\label{lemma: representation first addend second part}
	Let $t \geqslant 0$ and $\tilde P = P \otimes \hat{P}$. If $\tilde X \in L^1_{\tilde P}(\tilde \Omega)$, then 
	\begin{equation}\label{eq: representation first addend}
		\Ind{\tilde \tau \leqslant t} E^{\tilde P} [ \tilde X | \sigma(\tilde \tau) \vee \F_t] = \Ind{\tilde \tau \leqslant t} \left.E^{P} [ \varphi(x, \cdot) | \F_t]\right|_{x = \tilde \tau} \ \ \ \ \tilde P\text{-a.s.}, 
	\end{equation}
	where $\varphi$ is the measurable function
	\[
		\varphi: (\R_+ \times \Omega \ , \ \mathcal{B}(\R_+) \otimes \F^{P}_\infty) \rightarrow (\R, \mathcal{B}(\R)),
	\]
	such that 
	\begin{equation}\label{eq: representation X}
		\tilde X(\omega, \hat{\omega}) = \varphi (\tilde \tau(\omega, \hat{\omega}), \omega), \ \ \ \  (\omega, \hat{\omega}) \in \tilde{\Omega}.
	\end{equation}
\end{lemma}

\begin{proof}
	By Lemma \ref{lemma: representation first addend} and Remark \ref{rem: replacement F}, a unique representation (\ref{eq: representation X}) exists and the right-hand side of (\ref{eq: representation first addend}) is $\sigma(\tilde \tau) \vee \F_t$-measurable. We first show that relation (\ref{eq: representation first addend}) holds for indicator functions of a $\pi$-system which generates $\G = \sigma(\tilde \tau) \vee \F_\infty$. Given $s \geqslant 0$ and $A \in \F_\infty$, we show
	\begin{align*}
		\Ind{\tilde \tau \leqslant t} E^{\tilde P} [ \mathbf{1}_{\{\tilde \tau \leqslant s\}\cap\{A \times \hat{\Omega}\}} | \sigma(\tilde \tau) \vee \F_t] %=& \Ind{\tilde \tau \leqslant t} \Ind{\tilde \tau \leqslant s}  E^{\tilde P} [  \mathbf{1}_{A \times \hat{\Omega}}| \sigma(\tilde \tau) \vee \F_t]\\
		= \Ind{\tilde \tau \leqslant t} \Ind{\tilde \tau \leqslant s}  E^{P} [  \mathbf{1}_A| \F_t] \ \ \ \ \tilde P\text{-a.s.}
	\end{align*}
	%The last equality follows from 
	%\[
	%	E^{\tilde P} [ \Ind{\tilde \tau \leqslant t} \mathbf{1}_{A \times \hat{\Omega}}| \sigma(\tilde \tau) \vee \F_t] = \Ind{\tilde \tau \leqslant t} E^{P} [  \mathbf{1}_A| \F_t] \ \ \ \ \tilde P\text{-a.s.}
	%\]
	Indeed, let $u \geqslant 0$ and $B \in \F_t$, 
	\begin{align*}
		\int_{\{ \tilde \tau \leqslant u \} \cap \{B \times \hat{\Omega}\}} \Ind{\tilde \tau \leqslant t}\Ind{\tilde \tau \leqslant s} \mathbf{1}_{A \times \hat{\Omega}} \ud \tilde P =& \int_{B \times \hat{\Omega}} \Ind{\tilde \tau \leqslant t \wedge s \wedge u}  \mathbf{1}_{A \times \hat{\Omega}} \ud \tilde P \\
		=& \int_{B \times \hat{\Omega}} E^{\tilde P} [ \Ind{\tilde \tau \leqslant t \wedge s \wedge u}  \mathbf{1}_{A \times \hat{\Omega}}| \F_t]  \ud \tilde P \\
		=& \int_{B \times \hat{\Omega}} E^{\tilde P} [ \Ind{\tilde \tau \leqslant t \wedge s \wedge u} | \F_t] \  E^{\tilde P} [\mathbf{1}_{A \times \hat{\Omega}}| \F_t]  \ud \tilde P \\
		=& \int_{B \times \hat{\Omega}} E^{\tilde P} [ \Ind{\tilde \tau \leqslant t \wedge s \wedge u} | \F_t] \  E^{P} [\mathbf{1}_A  | \F_t]  \ud \tilde P \\
		=& \int_{B \times \hat{\Omega}} E^{\tilde P} [ \Ind{\tilde \tau \leqslant t \wedge s \wedge u} E^{P} [\mathbf{1}_A  | \F_t] | \F_t]  \ud \tilde P \\
		=& \int_{B \times \hat{\Omega}} \Ind{\tilde \tau \leqslant t \wedge s \wedge u}   E^{P} [\mathbf{1}_A | \F_t]  \ud \tilde P \\
		=& \int_{\{ \tilde \tau \leqslant u \} \cap \{B \times \hat{\Omega}\}} \Ind{\tilde \tau \leqslant t} \Ind{\tilde \tau \leqslant s}  E^{P} [\mathbf{1}_A | \F_t]  \ud \tilde P,
	\end{align*}
	where in the third equality we use the $\F_t$-conditional independence between $\H_t$ and $\F_\infty$, see pp.166 of \cite{Bie-Rut}. 
	Lemma \ref{lemma: representation first addend} together with the conditional monotone convergence yields that the set of bounded measurable functions $\tilde X \in L^1_{\tilde P}(\tilde \Omega)$, which satisfy relation (\ref{eq: representation first addend}), contains a monotone class. Hence by Monotone Class theorem, relation (\ref{eq: representation first addend}) holds for all bounded measurable functions $\tilde X \in L^1_{\tilde P}(\tilde \Omega)$. The result can be extended to every $\tilde X \in L^1_{\tilde P}(\tilde \Omega)$ by conditional monotone convergence theorem applied to $\tilde{X}^+$ and $\tilde{X}^-$ respectively, since every nonnegative measurable function is the pointwise limit of a sequence of nonnegative and nondecreasing simple functions.
\end{proof}

We note that the above results hold clearly also for $\tilde X$ which is $\G^{\P}$-measurable and nonnegative. A summary is given in the following proposition.

\begin{prop}\label{prop: summary decomposition}
	Let $t \geqslant 0$ and $\tilde P = P \otimes \hat{P}$. If $\tilde X \in L^1_{\tilde P}(\tilde \Omega)$ or $\tilde X$ is $\G^{\P}$-measurable and nonnegative, then 
	\[
		E^{\tilde P} [\tilde X | \G_t] =   \Ind{\tilde \tau \leqslant t} \left.E^{P} [ \varphi(x, \cdot) | \F_t]\right|_{x = \tilde \tau} + \Ind{\tilde \tau > t} e^{\Gamma_t} E^{P}[ E^{\hat{P}}[ \Ind{\tilde \tau > t} \tilde X] | \F_t] \ \ \ \ \tilde P\text{-a.s.}, 
	\]
	where $\varphi$ is the measurable function
	\[
		\varphi: (\R_+ \times \Omega \ , \ \mathcal{B}(\R_+) \otimes \F^{P}_\infty) \rightarrow (\R, \mathcal{B}(\R)),
	\]
	such that 
	\begin{equation}
		\tilde X(\omega, \hat{\omega}) = \varphi (\tilde \tau(\omega, \hat{\omega}), \omega), \ \ \ \  (\omega, \hat{\omega}) \in \tilde{\Omega}.
	\end{equation}
\end{prop} 

\begin{proof}
	It is sufficient to apply Lemma \ref{lemma: representation second addend} and Lemma \ref{lemma: representation first addend second part} to decomposition (\ref{eq: Bielecki decomposition}).
\end{proof}

Before we state the main results, we list some properties of upper semianalytic functions which we will use later.

\begin{lemma}\label{lemma: properties upper sem function}
	Let $X$, $Y$ be two Polish spaces.
	\begin{enumerate}
	\item  If $f : X \rightarrow Y$ is a Borel-measurable function and a set $A \subseteq X$ is analytic, then $f(A)$ is analytic. If a set $B \subseteq Y$ is analytic, then $f^{-1}(B)$ is analytic.
	
	\item If $f_n : X \rightarrow \bar{\R}$, $n \in \N$, is a sequence of upper semianalytic functions and $f_n \rightarrow f$, then $f$ is upper semianalytic.
	
	\item If $f: X \rightarrow Y$ is a Borel-measurable function and $g: Y \rightarrow \bar{\R}$ is upper semianalytic, then the composition $g \circ f$ is also upper semianalytic. If $f: X \rightarrow Y$ is a surjective Borel-measurable function and there is a function $g: Y \rightarrow \bar{\R}$ such that $g \circ f$ is upper semianalytic, then $g$ is upper semianalytic.
	
	\item If $f$, $g : X \rightarrow \bar{\R}$ are two upper semianalytic functions, then $f + g$ is upper semianalytic.
	
	\item If $f : X \rightarrow \bar{\R}$ is an upper semianalytic function, $g: X \rightarrow \bar{\R}$ is a Borel-measurable function and $g \geqslant 0$, then the product $f \cdot g$ is upper semianalytic.
	
	\item If $f: X \times Y \rightarrow \bar{\R}$ is upper semianalytic and $\kappa(\ud y; x)$ is a Borel-measurable stochastic kernel on $Y$ given $X$, then the function $g: X \rightarrow \bar{\R}$ defined by
	\[
		g(x) = \int f(x, y) \kappa (\ud y; x), \ \ \ \ x \in X,
	\]
	is upper semianalytic.
	\end{enumerate}
\end{lemma}

\begin{proof}
	See Proposition 7.40, Lemma 7.30 and Proposition 7.48 of \cite{Ber-Shr}\footnote{In \cite{Ber-Shr}, only lower semianalytic functions are considered. However, the results hold also for upper semianalytic functions without changes.} for points 1, 2, 4, 5 and 6. For the third point, the fact that $g$ upper semianalytic implies $g \circ f$ upper semianalytic is proved in Lemma 7.30 (3) of \cite{Ber-Shr}. For the inverse implication, we note that if $g \circ f$ is upper semianalytic, then for every $c \in \R$, the set 
	\[
		A := \{ x \in X : g \circ f (x)  > c \}
	\]
	is analytic. Moreover, if we define
	\[
		B := \{ y \in Y: g(y) > c \},
	\]
	we have $f(A) \subseteq B$. Since $f$ is surjective, it also holds that for all $y \in B$, there exists $x \in X$ such that $y = f(x)$ and $g(f(x)) > c$. Hence $f(A) \supseteq B$. It follows from the first point that the set $B$ is analytic. This implies that $g$ is upper semianalytic.
\end{proof}

\begin{theorem}\label{theo: definition G conditional expectation}
	Let Assumption \ref{assump: conditions on P (only one) for conditional sub exp} hold for $\P$ and consider an upper semianalytic function $\tilde X$ on $\tilde \Omega$ such that $\tilde X \in L^1(\tilde \Omega)$ or $\tilde X$ is $\G^{\P}$-measurable and nonnegative. If $t \geqslant 0$, then the following function
	\begin{equation}\label{eq: existence G conditional expectation}
		\mathcal{\tilde E}_{t}(\tilde X) := \Ind{\tilde \tau \leqslant t}\left.\mathcal{ E}_{t}( \varphi(x, \cdot))\right|_{x = \tilde \tau}  + \Ind{\tilde \tau > t}   \mathcal{ E}_{t}( e^{\Gamma_t} E^{\hat{P}}[ \Ind{\tilde \tau > t}\tilde X] )
	\end{equation}
	is well defined, where $\varphi$ is the measurable function
	\[
		\varphi: (\R_+ \times \Omega \ , \ \mathcal{B}(\R_+) \otimes \F^{\P}_\infty) \rightarrow (\R, \mathcal{B}(\R)),
	\]
	such that 
	\begin{equation*}
		\tilde X(\omega, \hat{\omega}) = \varphi (\tilde \tau(\omega, \hat{\omega}), \omega), \ \ \ \  (\omega, \hat{\omega}) \in \tilde{\Omega}.
	\end{equation*}
	Furthermore, $\mathcal{\tilde E}_{t}(\tilde X)$ satisfies the consistency condition (\ref{eq: consistency condition G-cond exp}). 
\end{theorem}

\begin{proof}
	By points 5 and 6 of Lemma \ref{lemma: properties upper sem function}, $e^{\Gamma_t} E^{\hat{P}}[ \Ind{\tilde \tau > t}\tilde X]$ is an upper semianalytic function on $\Omega$. Hence the second component on the right-hand side of (\ref{eq: existence G conditional expectation}) is well defined. For the first component, it is sufficient to prove that for every fixed $x \in \R_+$, the function $\varphi_x(\omega):= \varphi(x, \omega)$, $\omega \in \Omega$, is upper semianalytic. Firstly, $\varphi$ as function of $(x, \omega) \in \R_+ \times \Omega$ is upper semianalytic by Remark \ref{rem: tau surjective} and the second implication of point 3 of Lemma \ref{lemma: properties upper sem function}, since $\tilde X (\omega, \hat{\omega}) = \varphi \circ (\tau, id|_\Omega) (\omega, \hat{\omega})$, $(\omega, \hat{\omega}) \in \Omega \times \hat{\Omega}$ is upper semianalytic. Secondly, for every fixed $x \in \R_+$, by the first implication of point 3 of Lemma \ref{lemma: properties upper sem function} we have that $\varphi_x$ as function of $\omega \in \Omega$ is also upper semianalytic, since  $\varphi_x = \varphi \circ \psi_x$ where $\psi_x(\omega) := (x, \omega)$, $\omega \in \Omega$, and the function $\psi_x$ is Borel-measurable.
	
	Now we show that consistency condition (\ref{eq: consistency condition G-cond exp}) holds.
	By Proposition \ref{prop: construction F-cond sub exp (only one P)}, under every $\tilde P \in \tilde{\P}$ we have
	\[
		\Ind{\tilde \tau \leqslant t}\left.\mathcal{ E}_{t}( \varphi(x, \cdot))\right|_{x = \tilde \tau} = \Ind{\tilde \tau \leqslant t} \underset{P' \in \mathcal{P}(t; P)}{\text{ess sup}^{P}} \left.E^{P'} [ \varphi(x, \cdot) | \F_t]\right|_{x = \tilde \tau} \ \ \ \ \tilde P \text{-a.s.},
	\]
	\[
		\Ind{\tilde \tau > t}   \mathcal{ E}_{t}( e^{\Gamma_t} E^{\hat{P}}[ \Ind{\tilde \tau > t}\tilde X] ) =  \Ind{\tilde \tau > t} \underset{P' \in \mathcal{P}({t}; P)}{\text{ess sup}^{P}} E^{P'}[e^{\Gamma_t}  E^{\hat{P}}[ \Ind{\tilde \tau > t} \tilde X] | \F_t] \ \ \ \ \tilde P \text{-a.s.}
	\]
	Moreover, for every $\tilde P = P \otimes \hat P$,
	\[
		\mathcal{\tilde P}(t; \tilde P) = \{\tilde{P'} \in \tilde{\P} : P' \otimes \hat{P} = P \otimes \hat{P} \text{ on } \G_t \} = \{\tilde{P'} \in \tilde{\P} : {P'} =  P \text{ on } \F_t \}.
	\]
	Hence, $\tilde P$-a.s. we have that
	\[
		\underset{P' \in \mathcal{P}(t; P)}{\text{ess sup}^{P}} \left.E^{P'} [ \varphi(x, \cdot) | \F_t]\right|_{x = \tilde \tau} = \underset{\tilde P' \in \mathcal{\tilde P}(t; \tilde P)}{\text{ess sup}^{\tilde P}} \left.E^{P'} [ \varphi(x, \cdot) | \F_t]\right|_{x = \tilde \tau},
	\]
	\[
		\underset{P' \in \mathcal{P}({t}; P)}{\text{ess sup}^{P}} E^{P'}[  e^{\Gamma_t} E^{\hat{P}}[ \Ind{\tilde \tau > t} \tilde X] | \F_t] = \underset{\tilde P' \in \mathcal{\tilde P}({t}; \tilde P)}{\text{ess sup}^{\tilde P}} E^{P'}[e^{\Gamma_t}  E^{\hat{P}}[ \Ind{\tilde \tau > t} \tilde X] | \F_t].
	\]
	We note that $\{\tilde \tau \leqslant t\}$ and $\{\tilde \tau > t\}$ are disjoint events, hence $\tilde P$-a.s.
	\begin{align*}
		&\Ind{\tilde \tau \leqslant t}  \underset{\tilde P' \in \mathcal{\tilde P}(t; \tilde P)}{\text{ess sup}^{\tilde P}} \left.E^{P'} [ \varphi(x, \cdot) | \F_t]\right|_{x = \tilde \tau}\\
		& + \Ind{\tilde \tau > t} \underset{\tilde P' \in \mathcal{\tilde P}({t}; \tilde P)}{\text{ess sup}^{\tilde P}}e^{\Gamma_t}  E^{P'}[e^{\Gamma_t}  E^{\hat{P}}[ \Ind{\tilde \tau > t} \tilde X] | \F_t]\\
		=& \underset{\tilde P' \in \mathcal{\tilde P}(t; \tilde P)}{\text{ess sup}^{\tilde P}} \left( \Ind{\tilde \tau \leqslant t}\left.E^{P'} [ \varphi(x, \cdot) | \F_t]\right|_{x = \tilde \tau} + \Ind{\tilde \tau > t} E^{P'}[e^{\Gamma_t}  E^{\hat{P}}[ \Ind{\tilde \tau > t} \tilde X] | \F_t] \right).
	\end{align*}
	Finally, since the integrability conditions on $\tilde X$ guarantee that we can apply the Fubini-Tonelli Theorem, then Proposition \ref{prop: summary decomposition} yields 
	\[
		\mathcal{\tilde E}_{t}(\tilde X) = \underset{\tilde P' \in \mathcal{\tilde P}({t}; \tilde P)}{\text{ess sup}^{\tilde P}} E^{\tilde {P'}} [ \tilde X | \G_t ] \ \ \ \ \tilde P\text{-a.s. for all }\tilde P \in \tilde{\P}.
	\]
\end{proof}

\begin{rem}\label{rem: properties G conditional expectation}
	Set $t \geqslant 0$ and let $\tilde X$ satisfy the conditions in Theorem \ref{theo: definition G conditional expectation}. The following holds:
	\begin{enumerate}
		%\item Lemma \ref{lemma: representation first addend second part} and Proposition \ref{prop: definition G conditional expectation} do not depend on the choice of function $\varphi$ in representation (\ref{eq: representation X with phi}).

		\item If $\tilde X (\omega, \hat \omega) = X (\omega)$ for all $\hat \omega \in \hat \Omega$, then  $\tilde{\mathcal{E}}_t(X)$ defined in (\ref{eq: existence G conditional expectation}) coincides with $\mathcal{E}_t(X)$ defined in (\ref{eq: def cond expect}).
		
		\item The function $\tilde{\mathcal{E}}_t(\tilde X)$ defined in (\ref{eq: existence G conditional expectation}) is sublinear in $\tilde X$.
		
		\item If $\tilde{Y}$ is an upper semianalytic function on $\tilde \Omega$, such that $\tilde Y \in L^1(\tilde \Omega)$ and 
		\[
			\underset{\tilde P' \in \mathcal{\tilde P}({t}; \tilde P)}{\text{ess sup}^{\tilde P}} E^{\tilde {P'}} [ \tilde X | \G_t ] = \underset{\tilde P' \in \mathcal{\tilde P}({t}; \tilde P)}{\text{ess sup}^{\tilde P}} E^{\tilde {P'}} [ \tilde Y | \G_t ] \ \ \ \ \tilde P \text{-a.s. for all } \tilde P \in \tilde{\P},
		\]
		then $\tilde{\mathcal{E}}_t( \tilde X) = \tilde{\mathcal{E}}_t( \tilde{Y})$ $\tilde P$-a.s. for all $\tilde P \in \tilde{\P}$.
		
		\item If $A \in \G_t$, then  $\tilde{\mathcal{E}}_t(\mathbf{1}_A \tilde X) = \mathbf{1}_A \tilde{\mathcal{E}}_t( \tilde X)$. This follows from Lemma 5.1.1 of \cite{Bie-Rut} and the above point.
		
		\item The following pathwise equalities hold:
		\[
			\tilde{\mathcal{E}}_t(\Ind{\tilde \tau \leqslant t} \tilde X) = \Ind{\tilde \tau \leqslant t}\tilde{\mathcal{E}}_t( \tilde X),
		\]
		\[
			\tilde{\mathcal{E}}_t(\Ind{\tilde \tau > t} \tilde X) = \Ind{\tilde \tau > t}\tilde{\mathcal{E}}_t(\tilde X),
		\]
		\[
			\tilde{\mathcal{E}}_t(\tilde X) = \tilde{\mathcal{E}}_t(\Ind{\tilde \tau \leqslant t} \tilde X) + \tilde{\mathcal{E}}_t(\Ind{\tilde \tau > t} \tilde X).
		\]
	\end{enumerate}
\end{rem}

\begin{rem}
	We note that in Theorem \ref{theo: definition G conditional expectation}, integrability conditions are required on the upper semianalytic function $\tilde X$ in order to define the sublinear operator $\tilde{\mathcal{E}}_t$. These are necessary for applying Fubini-Tonelli Theorem in the proof. This creates a fundamental difference with respect to the construction in \cite{Nut-Han}, where  measurability conditions alone are sufficient for defining the sublinear operator in (\ref{eq: def cond expect}).
\end{rem}
\ \\ 
For the sake of simplicity, we use the following notations
\begin{equation}\label{eq: notation Et}
	E^{P}[\tilde X | \F_t] (\omega, \hat \omega): = E^{P}[\tilde X (\cdot, \hat \omega) | \F_t] (\omega), \ \ \ \ (\omega, \hat \omega) \in \tilde \Omega, \ t \geqslant 0,
\end{equation}
\begin{equation}\label{eq: notation mathcal Et}
	\mathcal{E}_t (\tilde X) (\omega, \hat \omega): = \mathcal{E}_t (\tilde X (\cdot, \hat \omega))(\omega), \ \ \ \ (\omega, \hat \omega) \in \tilde \Omega, \ t \geqslant 0.
\end{equation}
We note that since the concatenation function is Borel-measurable, the right-hand side of (\ref{eq: notation mathcal Et}) is well defined by (\ref{eq: def cond expect}) and points 3 and 6 of Lemma \ref{lemma: properties upper sem function}.

\begin{prop}\label{prop: integrability G-cond exp}
	Let Assumption \ref{assump: conditions on P (only one) for conditional sub exp} hold for $\P$ and let $\tilde X$ be an upper semianalytic function on $\tilde \Omega$ such that  $\tilde X \in L^1(\tilde \Omega)$ or $\tilde X$ is $\G^{\P}$-measurable and nonnegative. For every $t \geqslant 0$, the function $\mathcal{\tilde E}_{t}(\tilde X)$ defined in (\ref{eq: existence G conditional expectation}) is upper semianalytic and measurable with respect to $\G^*_t$ and $\G^{\P}$.
\end{prop}

\begin{proof}
	Let $t \geqslant 0$. By definition (\ref{eq: existence G conditional expectation}) and Proposition \ref{prop: construction F-cond sub exp (only one P)}, we have that $\mathcal{\tilde E}_{t}(\tilde X)$ is $(\F^*_t \vee \sigma(\tau))$-measurable, hence also $\G^*_t$- and $\G^{\P}$-measurable. It is upper semianalytic by points 3, 4, 5 of Lemma \ref{lemma: properties upper sem function} and Proposition \ref{prop: construction F-cond sub exp (only one P)}.
\end{proof}

Remark \ref{rem: properties G conditional expectation} and Proposition \ref{prop: integrability G-cond exp} show that $(\tilde{\mathcal{E}}_t)_{t\geqslant 0}$ is a family of sublinear conditional expectations which extends $({\mathcal{E}}_t)_{t\geqslant 0}$ defined for functions on $\Omega$.
We now prove that the family $(\tilde{\mathcal{E}}_t)_{t\geqslant 0}$ satisfies a weak form of dynamic programming principle or tower property, similar to the one of \cite{Peng}.

\begin{theorem}\label{theo: tower property}
	Let Assumption \ref{assump: conditions on P (only one) for conditional sub exp} hold and $\tilde X$ be an upper semianalytic function on $\tilde \Omega$ such that  $\tilde X$ is $\G^{\P}$-measurable and nonnegative. If $ 0 \leqslant s \leqslant t$, then
	\begin{equation}\label{eq: tower property}
		\mathcal{\tilde E}_{s}(\mathcal{\tilde E}_{t}(\tilde X)) \geqslant \mathcal{\tilde E}_{s}(\tilde X) \ \ \ \ \tilde P \text{-a.s. for all } \tilde P \in \tilde{\P}.
	\end{equation}
\end{theorem}

\begin{proof}
	We recall that we use notations (\ref{eq: notation E*}), (\ref{eq: notation Et}) and (\ref{eq: notation mathcal Et}). Since $\tilde X$ is assumed to be $\G^{\P}$-measurable and nonnegative, by Proposition \ref{prop: integrability G-cond exp} and the sublinearity of the operator $\tilde{\mathcal{E}}_t$, the left-hand side of (\ref{eq: tower property}) is well defined. By definition (\ref{eq: existence G conditional expectation}), relation (\ref{eq: tower property}) equals the following
	\begin{align}
		&\Ind{\tilde \tau \leqslant s}\left.\mathcal{ E}_{s}( \bar{\varphi}(x, \cdot))\right|_{x = \tilde \tau}  + \Ind{\tilde \tau > s}   \mathcal{ E}_{s}( e^{\Gamma_s} E^{\hat{P}}[ \Ind{\tilde \tau > s}\mathcal{\tilde E}_{t}(\tilde X)] ) \nonumber \\
		\geqslant&\Ind{\tilde \tau \leqslant s}\left.\mathcal{ E}_{s}( \varphi(x, \cdot))\right|_{x = \tilde \tau}  + \Ind{\tilde \tau > s}  \mathcal{ E}_{s}(e^{\Gamma_s}   E^{\hat{P}}[ \Ind{\tilde \tau > s}\tilde X] ),\label{eq: addends}
	\end{align}
	where $\varphi$ is the measurable function
	\[
		\varphi: (\R_+ \times \Omega \ , \ \mathcal{B}(\R_+) \otimes \F^{\P}_\infty) \rightarrow (\R, \mathcal{B}(\R)),
	\]
	such that 
	\begin{equation*}
		\tilde X(\omega, \hat{\omega}) = \varphi (\tilde \tau(\omega, \hat{\omega}), \omega), \ \ \ \  (\omega, \hat{\omega}) \in \tilde{\Omega},
	\end{equation*}
	and 
	\begin{align*}
		\bar \varphi (x, \omega) = \Ind{x \leqslant t}\mathcal{ E}_{t}( \varphi(x, \cdot))(\omega) + \Ind{x > t}   \mathcal{ E}_{t}(e^{\Gamma_t}  E^{\hat{P}}[ \Ind{\tilde \tau > t}\tilde X])(\omega),
	\end{align*}
	for all $(x, \omega) \in \R_+ \times \Omega$.
	We show first the equality between the first terms on both hand sides of (\ref{eq: addends}) by using (\ref{eq: existence G conditional expectation}) and the tower property (\ref{eq: tower property F}) of $(\P, \Fb)$-conditional expectation:
	\begin{align*}
		&\Ind{\tilde \tau \leqslant s}\left.\mathcal{ E}_{s}( \bar{\varphi}(x, \cdot))\right|_{x = \tilde \tau}\\
		=& \Ind{\tilde \tau \leqslant s}\left.\mathcal{ E}_{s}\left( \Ind{x \leqslant t}\mathcal{ E}_{t}( \varphi(x, \cdot)) + \Ind{x > t}   \mathcal{ E}_{t}(e^{\Gamma_t}  E^{\hat{P}}[ \Ind{\tilde \tau > t}\tilde X] )\right)\right|_{x = \tilde \tau}\\
		=& \Ind{\tilde \tau \leqslant s}\left.\left( \Ind{x \leqslant t}\mathcal{ E}_{s} \left( \mathcal{ E}_{t}( \varphi(x, \cdot)) \right) + \Ind{x > t} \mathcal{ E}_{s}\left( \mathcal{ E}_{t}(e^{\Gamma_t}   E^{\hat{P}}[ \Ind{\tilde \tau > t}\tilde X] )\right)\right)\right|_{x = \tilde \tau}\\
		=& \Ind{\tilde \tau \leqslant s}\left.\mathcal{ E}_{s} \left( \mathcal{ E}_{t}( \varphi(x, \cdot)) \right)\right|_{x = \tilde \tau}\\
		=&\Ind{\tilde \tau \leqslant s}\left.\mathcal{ E}_{s} ( \varphi(x, \cdot)) \right|_{x = \tilde \tau}.
	\end{align*}
	For the second terms, we note first that for every fixed $\hat \omega \in \hat \Omega$, $\tilde \tau(\cdot, \hat \omega)$ is an $\Fb$-stopping time. Hence by Galmarino's test, on the event $\{\tilde \tau \leqslant t \}$ we have
	\[
		\tilde \tau (\omega \otimes_t \omega', \hat \omega) = \tilde \tau (\omega, \hat \omega) \ \ \ \ \text{for all } \omega' \in \Omega.
	\]
	Hence on the event $\{ \tilde \tau \leqslant t \}$, for every fixed $\hat \omega \in \hat \Omega$,  by using definitions (\ref{eq: def concatenation}), (\ref{eq: def cond expect}) and representation (\ref{eq: representation X with phi}), we have
	\begin{align*}
		 \mathcal{E}_t (\tilde X)(\omega, \hat \omega) =& \sup_{P \in \P} \int_{\Omega} \tilde X (\omega \otimes_t \omega', \hat \omega) P(\ud \omega')\\
		 =&\sup_{P \in \P} \int_{\Omega} \varphi (\tilde \tau(\omega \otimes_t \omega', \hat \omega), \omega \otimes_t \omega') P(\ud \omega')\\
		 =&\sup_{P \in \P} \int_{\Omega} \varphi (\tilde \tau(\omega, \hat \omega), \omega \otimes_t \omega') P(\ud \omega')\\
		 =&\left.\sup_{P \in \P} \int_{\Omega} \varphi (x, \omega \otimes_t \omega') P(\ud \omega')\right|_{x = \tilde \tau (\omega, \hat \omega)}\\
		 =&\mathcal{E}_t (\varphi(x, \cdot ))(\omega)|_{x = \tilde \tau (\omega, \hat \omega)} \ \ \ \ \text{for all }\omega \in \Omega,
	\end{align*}
	that is
	\begin{equation}\label{eq: fixed hat omega}
		\Ind{\tilde \tau \leqslant t} \left.\mathcal{E}_t (\varphi(x, \cdot))\right|_{x = \tilde \tau} = \Ind{\tilde \tau \leqslant t} \mathcal{E}_t (\tilde X) \ \ \ \ \text{for every fixed }\hat \omega \in \hat \Omega.
	\end{equation}
	Furthermore, we note that by (\ref{eq: consist cond for t}), for every $P \in \P$
	\begin{equation}\label{eq: extract e^gamma}
		\mathcal{ E}_{t}(e^{\Gamma_t}   E^{\hat{P}}[ \Ind{\tilde \tau > t}\tilde X] ) = e^{\Gamma_t}  \mathcal{ E}_{t}(  E^{\hat{P}}[ \Ind{\tilde \tau > t}\tilde X] ) \ \ \ \ P \text{-a.s.}
	\end{equation}
	Now by (\ref{eq: existence G conditional expectation}), (\ref{eq: fixed hat omega}), (\ref{eq: extract e^gamma}) and Remark 2.4 (iii) of \cite{Nut-Han}, we have
	\begin{align*}
		&\mathcal{ E}_{s}( e^{\Gamma_s} E^{\hat{P}}[ \Ind{\tilde \tau > s}\mathcal{\tilde E}_{t}(\tilde X)] ) \\
		=&e^{\Gamma_s} \mathcal{ E}_{s}\left( E^{\hat{P}}\left[ \Ind{\tilde \tau > s}(\Ind{\tilde \tau \leqslant t}\left.\mathcal{E}_t (\varphi(x, \cdot))\right|_{x = \tilde \tau} + \Ind{\tilde \tau > t} \mathcal{ E}_{t}( e^{\Gamma_t} E^{\hat{P}}[ \Ind{\tilde \tau > t}\tilde X] )) \right]\right)  \\
		=&e^{\Gamma_s} \mathcal{ E}_{s}\left( E^{\hat{P}}\left[ \Ind{s < \tilde \tau \leqslant t}\mathcal{E}_t (\tilde X) + \Ind{\tilde \tau > t}  e^{\Gamma_t} \mathcal{ E}_{t}(  E^{\hat{P}}[ \Ind{\tilde \tau > t}\tilde X] ) \right] \right)\\
		=&e^{\Gamma_s}\mathcal{ E}_{s}\left( E^{\hat{P}}[ \Ind{s < \tilde \tau \leqslant t}\mathcal{E}_t (\tilde X)] +  E^{\hat{P}}[\Ind{\tilde \tau > t}  e^{\Gamma_t}  \mathcal{ E}_{t}( E^{\hat{P}}[ \Ind{\tilde \tau > t}\tilde X] ) ]\right) \ \ \ \ P \text{-a.s. for all } P \in \P.
	\end{align*}
 	Since $e^{\Gamma_t}  \mathcal{ E}_{t}(   E^{\hat{P}}[ \Ind{\tilde \tau > t}\tilde X] ) $ depends only on the first component $\omega$, using the definition of $\Gamma$ and (\ref{eq: notation E*}), it follows from Lemma \ref{lemma: equivalence set for tau} that
	\begin{align}  
		E^{\hat P} \left[ \Ind{\tilde \tau > t}  e^{\Gamma_t}  \mathcal{ E}_{t}(   E^{\hat{P}}[ \Ind{\tilde \tau > t}\tilde X] ) \right]
		=& E^{\hat{P}}[\Ind{\tilde \tau > t} ] e^{\Gamma_t}  \mathcal{ E}_{t}(   E^{\hat{P}}[ \Ind{\tilde \tau > t}\tilde X] )  \nonumber \\
		=&  e^{-\Gamma_t} e^{\Gamma_t}  \mathcal{ E}_{t}(   E^{\hat{P}}[ \Ind{\tilde \tau > t}\tilde X] ) \nonumber \\
		=&  \mathcal{ E}_{t}(   E^{\hat{P}}[ \Ind{\tilde \tau > t}\tilde X] ). 
	\end{align}
	It follows
	\begin{align}
		\mathcal{ E}_{s}( e^{\Gamma_s} E^{\hat{P}}[ \Ind{\tilde \tau > s}\mathcal{\tilde E}_{t}(\tilde X)] ) 
		=&e^{\Gamma_s}\mathcal{ E}_{s}\left( E^{\hat{P}}[ \Ind{s < \tilde \tau \leqslant t}\mathcal{E}_t (\tilde X)] + \mathcal{ E}_{t}( E^{\hat{P}}[ \Ind{\tilde \tau > t}\tilde X] ) \right) \nonumber \\
		=&e^{\Gamma_s}\mathcal{ E}_{s}\left( E^{\hat{P}}[ \mathcal{E}_t (\Ind{s < \tilde \tau \leqslant t}\tilde X)] + \mathcal{ E}_{t}( E^{\hat{P}}[ \Ind{\tilde \tau > t}\tilde X] ) \right) \nonumber \\
		\geqslant &e^{\Gamma_s}\mathcal{ E}_{s}\left( \mathcal{E}_t ( E^{\hat{P}}[  \Ind{s < \tilde \tau \leqslant t}\tilde X]) + \mathcal{ E}_{t}(E^{\hat{P}}[ \Ind{\tilde \tau > t}\tilde X] )\right ) \label{eq: first inequality}\\
		\geqslant &e^{\Gamma_s}\mathcal{ E}_{s}\left( \mathcal{E}_t ( E^{\hat{P}}[  \Ind{s < \tilde \tau \leqslant t}\tilde X] + E^{\hat{P}}[ \Ind{\tilde \tau > t}\tilde X] ) \right) \label{eq: second inequality}\\
		=& e^{\Gamma_s}\mathcal{ E}_{s}( \mathcal{E}_t ( E^{\hat{P}}[ \Ind{\tilde \tau > s} \tilde X]  ) ) \nonumber \\
		=& e^{\Gamma_s}\mathcal{ E}_{s}(  E^{\hat{P}}[ \Ind{\tilde \tau > s} \tilde X]  ) \nonumber \\
		=& \mathcal{ E}_{s}(e^{\Gamma_s} E^{\hat{P}}[ \Ind{\tilde \tau > s} \tilde X]  ) \ \ \ \ P \text{-a.s. for all } P \in \P. \nonumber
	\end{align}
	In the second equality we use the properties that for every fixed $\hat \omega \in \hat \Omega$, $\{s < \tilde \tau (\cdot, \hat \omega) \leqslant t\} \in \F_t$ and $\mathcal{ E}_{t}(\mathbf{1}_A X) = \mathbf{1}_A \mathcal{ E}_{t}(X)$, if $A \in \F_t$ and $X$ is upper semianalytic, see Remark 2.4 (iv) of \cite{Nut-Han}.
	The inequality (\ref{eq: first inequality}) follows from (\ref{eq: consist cond for t}) and the conditional Fubini-Tonelli Theorem. Indeed, with notation (\ref{eq: notation E*}) we have
	\begin{align*}
		E^{\hat{P}}[ \mathcal{E}_t (\Ind{s < \tilde \tau \leqslant t}\tilde X)] &= E^{\hat{P}}\left[ \underset{ P' \in \mathcal{ P}({t};  P)}{\text{ess sup}^{ P}} \E^{{P'}}[\Ind{s < \tilde \tau \leqslant t} \tilde X | \mathcal{F}_{t}]\right]\\
		&\geqslant E^{\hat{P}}[ \E^{{P}}[\Ind{s < \tilde \tau \leqslant t} \tilde X | \mathcal{F}_{t}]] \\
		&= \E^{{P}}[ E^{\hat{P}}[ \Ind{s < \tilde \tau \leqslant t} \tilde X ] | \mathcal{F}_{t}] \ \ \ \ P \text{-a.s. for all } P \in \P.
	\end{align*}
	Hence,
	\begin{align*}
		E^{\hat{P}}[ \mathcal{E}_t (\Ind{s < \tilde \tau \leqslant t}\tilde X)] &\geqslant \underset{ P' \in \mathcal{ P}({t};  P)}{\text{ess sup}^{ P}} \E^{{P'}}[ E^{\hat{P}}[\Ind{s < \tilde \tau \leqslant t}\tilde X] | \mathcal{F}_{t}]\\
		&= \mathcal{E}_t ( E^{\hat{P}}[ \Ind{s < \tilde \tau \leqslant t}\tilde X]) \ \ \ \ P \text{-a.s. for all } P \in \P.
	\end{align*}
	The inequality (\ref{eq: second inequality}) follows from the sublinearity of $(\P, \Fb)$-conditional expectation.
	In the second last equality we use the tower property (\ref{eq: tower property F}). This concludes the proof.
\end{proof}

\begin{cor}\label{cor: tower property}
	Let Assumption \ref{assump: conditions on P (only one) for conditional sub exp} hold and $\tilde X$ be an upper semianalytic function on $\tilde \Omega$ such that  $\tilde X \in L^1(\tilde \Omega)$. If for $t \geqslant 0$, $\mathcal{\tilde E}_{t}(\tilde X) \in L^1(\tilde \Omega)$, then
	\begin{equation*}
		\mathcal{\tilde E}_{s}(\mathcal{\tilde E}_{t}(\tilde X)) \geqslant \mathcal{\tilde E}_{s}(\tilde X) \ \ \ \ \tilde P \text{-a.s. for all } \tilde P \in \tilde{\P},
	\end{equation*}
	for $ 0 \leqslant s \leqslant t$.
\end{cor}

In Appendix \ref{app: counterexample}, an explicit counterexample shows that the above weak tower property of the family $(\tilde{\mathcal{E}}_t)_{t \geqslant 0}$ cannot be improved in full generality. However, in Section \ref{sec: main products} we show that the classic tower property holds in all cases of practical interest for credit or insurance products. In Appendix \ref{app: other condition tower property}, further sufficient conditions for the tower property are provided.

\begin{rem}
	The classic dynamic programming property fails in the reduced-form setting due to the nature of the progressively enlarged filtration $\Gb$. Indeed, while the canonical filtration $\Fb$ is consistent with the 'path-pasting' construction shown in \cite{Nut-Han}, from which the dynamic programming property follows as a natural consequence, this is not the case for the enlarged filtration $\Gb$. Furthermore, we note that $\tilde{\mathcal{E}}_t$ does not always map $L^1(\tilde \Omega)$ into $L^1(\tilde \Omega)$,  the reason is the same that causes the dynamic programming property to fail. For a detailed discussion on these technical difficulties, we refer to \cite{Zhang}.
\end{rem}

In view of the above results, we give the following definition which extends the one in Proposition \ref{prop: construction F-cond sub exp (only one P)} to the reduced-form setting under model uncertainty.

\begin{defn}
	We call the family of sublinear conditional expectations $(\tilde{\mathcal{E}}_t)_{t\geqslant 0}$ \emph{$(\tilde{\P}, \Gb)$-conditional expectation}.
\end{defn}

\subsection{Valuation of credit and insurance products under model uncertainty}\label{sec: main products}

We now consider the valuation of credit and insurance products under model uncertainty. We show in Proposition \ref{prop: tower property main products} that in these cases, the classic tower property holds and the sublinear operator $\tilde{\mathcal{E}}_t$ maps $L^1(\tilde \Omega)$ into $L^1(\tilde \Omega)$. As we will see in Section \ref{sec: superhedging reduced form}, the following valuation formulas can be hence interpreted as superhedging prices for the given cash flows.

Let $T < \infty$ be the maturity time. 
We define the filtration $\Fb^{\P} := (\F^{\P}_t)_{t \in [0,T]}$ by
\[
	\F^{\P}_t := \F^*_t \vee \mathcal{N}_T^{\P}, \ \ \ \  t \in [0,T],
\]
where $\mathcal{N}_T^{\P}$ is the collection of sets which are $(P, \F_T)$-null for all $P \in \P$.
For both credit and insurance markets, the main products associated to a particular default event represented by $\tilde \tau$ can be modelled by three kinds of contracts with the following payoff:
\begin{enumerate}
	\item $\Ind{\tilde \tau > T} Y$, where $Y$ is an $\F^{\P}_T$-measurable nonnegative upper semianalytic function on $\Omega$ such that $\mathcal{E}(Y) < \infty$; i.e. the payment is made at the maturity of the contract only if the default event does not occur before the maturity date;
	
	\item $\Ind{0 < \tilde \tau \leqslant T } Z_{\tilde \tau}$, where $Z := (Z_t)_{t \in [0,T]}$ is an $\Fb^{\P}$-predictable nonnegative process on $\Omega$, such that the function $Z(t, \omega) := Z_t(\omega), (t, \omega) \in [0,T] \times \Omega$, is upper semianalytic  and $\sup_{t \in [0,T]}\mathcal{E}(Z_t) < \infty$; i.e. the payment is made at $\tilde \tau$ only if the default event occurs before or at the maturity of the contract;
	
	\item $\int_0^T (1 - H_u ) \ud C_u\footnote{This integral is a pathwisely defined Lebesgue-Stieltjes integral.} = \Ind{\tilde \tau > T}C_T  + \Ind{0 < \tilde \tau \leqslant T }C_{\tilde \tau-} $, where $C := (C_t)_{t \in [0,T]}$ is a nonnegative $\Fb^{\P}$-adapted nondecreasing process on $\Omega$, with $C(t, \omega): = C_t(\omega), (t, \omega) \in [0,T] \times \Omega$, upper semianalytic and $\sup_{t \in [0,T]}\mathcal{E}(C_t) < \infty$, which represents the cumulative payment; i.e. a payment flow is made as long as the default event does not occur or the contract is valid.
\end{enumerate}

\noindent We give first valuation formulas for these three kinds of contracts under model uncertainty.

\begin{lemma}\label{lemma: representation first product}
Let $Y = Y(\omega)$, $\omega \in \Omega$, be an $\F_T^{\P}$-measurable upper semianalytic function such that $\mathcal{E}(|Y|) < \infty$. Then for every $t \in [0,T]$,
\[
	\Ind{\tilde \tau > T} Y \ \ \ \ \text{and} \ \ \ \ Y e^{- \int^T_t \mu_u \ud u}
\]
are upper semianalytic and belong to $L^1(\tilde \Omega)$. Furthermore, if $\P$ satisfies Assumption \ref{assump: conditions on P (only one) for conditional sub exp}, the following holds pathwisely for every $t \in [0,T]$,
\begin{equation} 
	\mathcal{\tilde E}_t \left( \Ind{\tilde \tau > T} Y \right) = \Ind{\tilde \tau > t} \mathcal{E}_t \left( Y e^{- \int^T_t \mu_u \ud u} \right). \label{eq: first product representation}
\end{equation}
\end{lemma}

\begin{proof}
	We note that $\Ind{\tilde \tau > T}$ and $e^{- \int^T_t \mu_u \ud u}$ are nonnegative Borel-measurable functions. By point 5 of Lemma \ref{lemma: properties upper sem function} we have that
	\[
	\Ind{\tilde \tau > T} Y \ \ \ \ \text{and} \ \ \ \ Y e^{- \int^T_t \mu_u \ud u}
	\]
	are upper semianalytic and clearly belong to $L^1(\tilde \Omega)$. Equality (\ref{eq: first product representation}) follows from (\ref{eq: existence G conditional expectation}) and the fact that $Y$ does not depend on $\hat \omega \in \hat \Omega$,
	\begin{align*}
		\mathcal{\tilde E}_{t}(\Ind{\tilde \tau > T} Y) =& \Ind{\tilde \tau > t} \mathcal{E}_{t}( e^{\Gamma_t} E^{\hat P}[\Ind{\tilde \tau > T} Y])\\
		=& \Ind{\tilde \tau > t} \mathcal{E}_{t}( Y e^{\Gamma_t -\Gamma_T})\\
		=& \Ind{\tilde \tau > t} \mathcal{E}_{t}( Y e^{- \int^T_t \mu_u \ud u}).
	\end{align*}

\end{proof}

\begin{lemma}\label{lemma: G-cond ind + Y}
	Let $Z := (Z_t)_{t \in [0,T]}$ be an $\Fb^{\P}$-predictable process on $\Omega$. Then under every $\tilde P \in \tilde{\P}$ with $\tilde P = P \otimes \hat{P}$, we have
	\begin{equation}\label{eq: second product}
		E^{\tilde P} \left. \left[\Ind{s < \tilde \tau \leqslant t} Z_{\tilde \tau} \right| \G_s  \right] = \Ind{\tilde \tau > s} E^{P} \left. \left[ \int^t_s Z_u e^{- \int^u_s \mu_v \ud v} \mu_u \ud u \right| \F_s  \right] \ \ \ \ \tilde P\text{-a.s.},
	\end{equation}
	for $s, t \in [0,T]$ with $s \leqslant t$.
\end{lemma}

\begin{proof}
	Let $\tilde{P} \in \tilde{\P}$ and $0 \leqslant s \leqslant t\leqslant T$. By Proposition \ref{prop: gamma aggregator}, Proposition 5.1.1 and Corollary 5.1.3 of \cite{Bie-Rut}, which hold without the usual conditions on the filtrations, we have
	\[
		E^{\tilde P} \left. \left[\Ind{s < \tilde \tau \leqslant t} Z_{\tilde \tau} \right| \G_s  \right] = \Ind{\tilde \tau > s} E^{\tilde P} \left. \left[ \int^t_s Z_u e^{- \int^u_s \mu_v \ud v} \mu_u \ud u \right| \F_s  \right] \ \ \ \ \tilde P\text{-a.s.}
	\]
	Then $\tilde P$-a.s. equality (\ref{eq: second product}) follows from $P \otimes \hat{P} |_{({\Omega}, \F)} = P$.
\end{proof}

\begin{cor}\label{cor: construction G-cond sub exp}
	Let $Z := (Z_t)_{t \in [0,T]}$ be an $\Fb^{\P}$-predictable process on $\Omega$ such that the function $Z(t, \omega) := Z_t(\omega), (t, \omega) \in [0,T] \times \Omega$, is upper semianalytic and 
	\[
		\sup_{t \in [0,T]}\mathcal{E}(|Z_t|) < \infty.
	\]
	Then,
	\[
		\Ind{s < \tilde \tau \leqslant t} Z_{\tilde \tau} \ \ \ \ \text{and} \ \ \ \ \int^t_s Z_u e^{- \int^u_s \mu_v \ud v} \mu_u \ud u 
	\]
	are upper semianalytic and belong to $L^1(\tilde \Omega)$, for all $s, t \in [0,T]$ with $s \leqslant t$. 
	Furthermore, if Assumption \ref{assump: conditions on P (only one) for conditional sub exp} holds for $\P$, we have
	\begin{equation}
		\mathcal{\tilde E}_s \left( \Ind{s < \tilde \tau \leqslant t} Z_{\tilde \tau} \right) = \Ind{\tilde \tau > s} \mathcal{E}_s \left( \int^t_s Z_u e^{- \int^u_s \mu_v \ud v} \mu_u \ud u \right) \ \ \ \ \tilde P \text{-a.s. for all } \tilde P \in \tilde{\P},\label{eq: second product representation}
	\end{equation}
	for all $s, t \in [0,T]$ with $s \leqslant t$.
	
	If in addition $Z$ is a stepwise $\Fb$-predictable process, that is 
	\[
		Z_t = \sum_{i=0}^n Z_{t_i} \Ind{t_i < t \leqslant t_{i+1}}, \ \ \ \ t \in [0,T],
	\]
	where $t_0 = s < \cdot \cdot \cdot < t_{n + 1} = t$, $Z_{t_i}$ is $\F_{t_i}$-measurable for all $i = 0, ..., n$, then equality (\ref{eq: second product representation}) holds pathwisely, that is 
	\begin{equation}
		\mathcal{\tilde E}_s \left( \Ind{s < \tilde \tau \leqslant t} Z_{\tilde \tau} \right) = \Ind{\tilde \tau > s} \mathcal{E}_s \left( \int^t_s Z_u e^{- \int^u_s \mu_v \ud v} \mu_u \ud u \right). \label{eq: second product representation stepwise}
	\end{equation}
\end{cor}

\begin{proof}
	We note that point 6 of Lemma \ref{lemma: properties upper sem function} holds also for $Y = [0,T]$, $\kappa (\ud y; x) \equiv \ud y$. This together with points 3 and  5 of Lemma \ref{lemma: properties upper sem function} shows that
	\[
		\Ind{s < \tilde \tau \leqslant t} Z_{\tilde \tau} \ \ \ \ \text{and} \ \ \ \ \int^t_s Z_u e^{- \int^u_s \mu_v \ud v} \mu_u \ud u 
	\]
	are upper semianalytic and belong to $L^1(\tilde \Omega)$. Equality (\ref{eq: second product representation}) follows from Lemma \ref{lemma: G-cond ind + Y} and point 3 of Remark \ref{rem: properties G conditional expectation}.
	
	If $Z$ is a stepwise $\Fb$-predictable process, by (\ref{eq: existence G conditional expectation}) we have
	\begin{align*}
		\mathcal{\tilde E}_s \left( \Ind{s < \tilde \tau \leqslant t} Z_{\tilde \tau} \right) &= \Ind{\tilde \tau > s} \mathcal{E}_s \left( e^{\Gamma_s} E^{\hat P} \left[ \sum_{i=0}^n Z_{t_i} \Ind{t_i < \tilde \tau \leqslant t_{i+1}} \right] \right)\\
		&= \Ind{\tilde \tau > s} \mathcal{E}_s \left( e^{\Gamma_s} \sum_{i=0}^n Z_{t_i} E^{\hat P} \left[  \Ind{t_i < \tilde \tau \leqslant t_{i+1}} \right] \right)\\
		&= \Ind{\tilde \tau > s} \mathcal{E}_s \left( e^{\Gamma_s} \sum_{i=0}^n Z_{t_i} (e^{-\Gamma_{t_i} } - e^{-\Gamma_{t_{i+1}} }) \right)\\
		&= \Ind{\tilde \tau > s} \mathcal{E}_s \left( \int_s^t  Z_u e^{\Gamma_s - \Gamma_u} \ud \Gamma_u \right)\\
		&= \Ind{\tilde \tau > s} \mathcal{E}_s \left( \int^t_s Z_u e^{- \int^u_s \mu_v \ud v} \mu_u \ud u \right),
	\end{align*}
	where the integrals above are pathwise Lebesgue--Stieltjes integrals.
\end{proof}

\begin{lemma}\label{lemma: third product}
	Let $C := (C_t)_{t \in [0,T]}$ be a nonnegative $\Fb^{\P}$-adapted nondecreasing and continuous process on $\Omega$. Then under every $\tilde P \in \tilde{\P}$ with $\tilde P = P \otimes \hat{P}$, we have
	\begin{align}
		&E^{\tilde P} \left. \left[\int_s^t (1 - H_u) \ud C_u \right| \G_s  \right] \nonumber \\
		=& \Ind{\tilde \tau > s} E^{P} \left. \left[ \int^t_s C_{u} e^{- \int^u_s \mu_v \ud v} \mu_u \ud u + C_t e^{- \int_s^t \mu_u \ud u} \right| \F_s  \right] \ \ \ \ \tilde P\text{-a.s.},\label{eq: third product one P}
	\end{align}
	for all $s, t \in [0,T]$ with $s \leqslant t$.
\end{lemma}

\begin{proof}
	Let $\tilde{P} \in \tilde{\P}$ and $0 \leqslant s \leqslant t\leqslant T$. We use the same proof of the first part of Proposition 5.1.2 of \cite{Bie-Rut}, which hold without the usual conditions on the filtrations, together with Proposition \ref{prop: gamma aggregator} and get
	\begin{align*}
		&E^{\tilde P} \left. \left[\int_s^t (1 - H_u) \ud C_u \right| \G_s  \right] \nonumber \\
		=& \Ind{\tilde \tau > s} E^{\tilde P} \left. \left[ \int^t_s C_{u} e^{- \int^u_s \mu_v \ud v} \mu_u \ud u + C_t e^{- \int_s^t \mu_u \ud u} \right| \F_s  \right] \ \ \ \ \tilde P\text{-a.s.}
	\end{align*}
	Then $\tilde P$-a.s. equality (\ref{eq: third product one P}) follows from $P \otimes \hat{P} |_{({\Omega}, \F)} = P$.
\end{proof}

\begin{cor}\label{cor: third product}
	Let $C := (C_t)_{t \in [0,T]}$ be a nonnegative $\Fb^{\P}$-adapted nondecreasing process on $\Omega$, with $C(t, \omega) := C_t(\omega)$, $(t ,\omega) \in [0,T] \times \Omega$ and upper semianalytic and $\sup_{t \in [0,T]}\mathcal{E}(C_t) < \infty$. Then 
	\[
		\int_s^t (1 - H_u) \ud C_u \ \ \ \ \text{and} \ \ \ \ \int^t_s C_{u} e^{- \int^u_s \mu_v \ud v} \mu_u \ud u + C_t e^{- \int_s^t \mu_u \ud u}
	\]
	are upper semianalytic and belong to $L^1(\tilde \Omega)$ for all $s, t \in [0,T]$ with $s \leqslant t$.
	Furthermore, if Assumption \ref{assump: conditions on P (only one) for conditional sub exp} holds for $\P$, we have
	\begin{align}
		&\mathcal{\tilde E}_s \left( \int_s^t (1 - H_u) \ud C_u \right) \nonumber \\
		=& \Ind{\tilde \tau > s} \mathcal{E}_s \left( \int^t_s C_{u} e^{- \int^u_s \mu_v \ud v} \mu_u \ud u + C_t e^{- \int_s^t \mu_u \ud u} \right) \ \ \ \ \tilde P \text{-a.s. for all } \tilde P \in \tilde{\P},\label{eq: third product representation}
	\end{align}
	for all $s, t \in [0,T]$ with $s \leqslant t$.
\end{cor}

\begin{proof}
	Since 
	\[
		\int_s^t (1 - H_u) \ud C_u = \Ind{s < \tilde \tau \leqslant t }C_{\tilde \tau} + \Ind{\tilde \tau > t}C_t,
	\]
	points 2, 4, 5 and 6 of Lemma \ref{lemma: properties upper sem function} show that 
	\[
		\int_s^t (1 - H_u) \ud C_u \ \ \ \ \text{and} \ \ \ \ \int^t_s C_{u} e^{- \int^u_s \mu_v \ud v} \mu_u \ud u + C_t e^{- \int_s^t \mu_u \ud u}
	\]
	are upper semianalytic and belong to $L^1(\tilde \Omega)$. Equality (\ref{eq: third product representation}) follows from Lemma \ref{lemma: third product} and point 3 of Remark \ref{rem: properties G conditional expectation}.
\end{proof}

Now we show that in all these cases of practical interest, the classic tower property holds and the sublinear operator $\tilde{\mathcal{E}}_t$ maps $L^1(\tilde \Omega)$ into $L^1(\tilde \Omega)$. The following proposition is slightly more general.

\begin{prop}\label{prop: tower property main products}
	Let $Z := (Z_t)_{t \in [0,T]}$ be an $\Fb^{\P}$-predictable process on $\Omega$ such that $\sup_{t \in [0,T]}\mathcal{E}(|Z_t|) < \infty$ and the function $Z(t, \omega) := Z_t(\omega), (t, \omega) \in [0,T] \times \Omega$, is upper semianalytic, and $Y$  an $\G^{\P}$-measurable upper semianalytic function on $\Omega$ such that $\mathcal{E}(|Y|) < \infty$. Let Assumption \ref{assump: conditions on P (only one) for conditional sub exp} hold for $\P$. If
\[
	\tilde X = \Ind{0 < \tau \leqslant T} Z_{\tilde \tau} + \Ind{\tilde \tau > T} Y,
\]	
then we have
\[
	\mathcal{\tilde E}_{t}(\tilde X) \in L^1(\tilde \Omega),
\]
for all $t \in [0,T]$
and	 the tower property holds, i.e.
	\[
		\mathcal{\tilde E}_{s}(\mathcal{\tilde E}_{t}(\tilde X)) = \mathcal{\tilde E}_{s}(\tilde X) \ \ \ \ \tilde P \text{-a.s. for all } \tilde P \in \tilde{\P},
	\]
	for all $s, t \in [0,T]$ with $s \leqslant t$. 
\end{prop}

\begin{proof}
	Let $t \in [0,T]$. Arguments similar to Lemma \ref{lemma: representation first product} and Corollary \ref{cor: construction G-cond sub exp} show that $\mathcal{\tilde E}_{t}(\tilde X)$ is well defined and $\mathcal{\tilde E}(|\tilde X|) < \infty$. We prove first that 
	\[
		\mathcal{\tilde E} ( |\mathcal{\tilde E}_{t}(\tilde X)|) < \infty.
	\]
	By calculations similar to the ones in Theorem \ref{theo: tower property} and Corollary \ref{cor: construction G-cond sub exp}, we have 
	\begin{align}
		&\sup_{\tilde P \in \tilde \P} E^{\tilde P} \left[ \left| \mathcal{\tilde E}_{t}(\tilde X) \right| \right] \nonumber \\
		\leqslant& \sup_{\tilde P \in \tilde \P} E^{\tilde P} \left[ \left| \Ind{\tilde \tau \leqslant t}\left.\mathcal{ E}_{t}( \varphi(x, \cdot))\right|_{x = \tilde \tau}\right| \right]  + \sup_{\tilde P \in \tilde \P} E^{\tilde P} \left[ \left|\Ind{\tilde \tau > t}   \mathcal{ E}_{t} \left(e^{\Gamma_t}  E^{\hat{P}}[ \Ind{\tilde \tau > t}  \tilde X] \right)  \right| \right] \label{eq: 1}\\
		=& \sup_{\tilde P \in \tilde \P} E^{\tilde P} \left[ \left| \Ind{s < \tilde \tau \leqslant t}Z_{\tilde \tau} \right| \right]  + \sup_{P \in \P} E^{P} \left[ E^{\hat P} \left[ \left|\Ind{\tilde \tau > t}   \mathcal{ E}_{t} \left(e^{\Gamma_t}  E^{\hat{P}}[ \Ind{\tilde \tau > T} Y] \right)  \right| \right] \right] \nonumber \\
		= &\sup_{P \in \P} E^{P} \left[ E^{\hat P}  \left[ \left| \Ind{s < \tilde \tau \leqslant t}Z_{\tilde \tau} \right| \right] \right]  + \sup_{P \in \P} E^{P} \left[  \left|   \mathcal{ E}_{t}( E^{\hat{P}}[ \Ind{\tilde \tau > T} Y] ) \right| \right] \nonumber \\
		\leqslant & \sup_{P \in \P} E^{P} \left[ \int^t_s |Z_u| e^{- \Gamma_u} \ud \Gamma_u \right] + \sup_{P \in \P} E^{P} \left[  \left|  E^{\hat{P}}[ \Ind{\tilde \tau > T} Y]  \right| \right] \label{eq: 2}\\
		\leqslant & \int^t_s \sup_{P \in \P} E^{P} [ |Z_u|] e^{- \Gamma_u} \ud \Gamma_u   + \sup_{P \in \P} E^{P} \left[   |Y| \right]\nonumber \\
		< & \ \infty, \nonumber
	\end{align}
	where (\ref{eq: 1}) is a consequence of the definition (\ref{eq: existence G conditional expectation}) and (\ref{eq: 2}) follows from Step 1 of the proof of Theorem 2.3 in \cite{Neu-Nut-sup} applied to the second component. This shows that for every $t > 0$, $\tilde{\mathcal{E}}_t(\tilde{X})$ still belongs to $L^1(\tilde \Omega)$.
	We now prove the tower property. Let $\tilde P \in \tilde \P$, by the proof of Theorem \ref{theo: tower property}, the classic tower property holds if and only if (\ref{eq: first inequality}) and (\ref{eq: second inequality}) are equalities. That is
\begin{align*}
	e^{\Gamma_s}\mathcal{ E}_{s}\left( E^{\hat{P}}[ \mathcal{E}_t (\Ind{s < \tilde \tau \leqslant t}\tilde X)] + \mathcal{ E}_{t}( E^{\hat{P}}[ \Ind{\tilde \tau > t}\tilde X] ) \right) = e^{\Gamma_s}\mathcal{ E}_{s}( \mathcal{E}_t ( E^{\hat{P}}[ \Ind{\tilde \tau > s} \tilde X]  ) ) \ \ \ \ \tilde P\text{-a.s.}
\end{align*}	
	
\noindent We have indeed
\begin{align*}
	& E^{\hat{P}}[ \mathcal{E}_t (\Ind{s < \tilde \tau \leqslant t}\tilde X)] + \mathcal{ E}_{t}( E^{\hat{P}}[ \Ind{\tilde \tau > t}\tilde X] )  \nonumber \\
		= &E^{\hat{P}}[ \mathcal{E}_t  ( \Ind{s < \tilde \tau \leqslant t}Z_{\tilde \tau})] + \mathcal{ E}_{t}(E^{\hat{P}}[ \Ind{t < \tilde \tau \leqslant T}Z_{\tilde \tau} + \Ind{\tilde \tau > T} Y] ) \\
		= &E^{\hat{P}}[ \Ind{s < \tilde \tau \leqslant t}Z_{\tilde \tau}] + \mathcal{ E}_{t}(E^{\hat{P}}[ \Ind{t < \tilde \tau \leqslant T}Z_{\tilde \tau} + \Ind{\tilde \tau > T} Y] ) \\
		= & \int^t_s Z_u e^{- \Gamma_u} \ud \Gamma_u + \mathcal{ E}_{t}(E^{\hat{P}}[ \Ind{t < \tilde \tau \leqslant T}Z_{\tilde \tau} + \Ind{\tilde \tau > T} Y] )   \\
		= & \mathcal{ E}_{t} \left(\int^t_s Z_u e^{- \Gamma_u} \ud \Gamma_u + E^{\hat{P}}[  \Ind{t < \tilde \tau \leqslant T}Z_{\tilde \tau} + \Ind{\tilde \tau > T} Y] \right)   \\
		= & \mathcal{ E}_{t} \left(E^{\hat{P}}[ \Ind{s < \tilde \tau \leqslant t}Z_{\tilde \tau}]+ E^{\hat{P}}[  \Ind{t < \tilde \tau \leqslant T}Z_{\tilde \tau} + \Ind{\tilde \tau > T} Y] \right)   \\
		=& \mathcal{E}_t ( E^{\hat{P}}[ \Ind{\tilde \tau > s} \tilde X]  )\ \ \ \ \tilde P\text{-a.s.},
\end{align*}	
where we stress that for fixed $\hat \omega$,  $\Ind{s < \tilde \tau \leqslant t}Z_{\tilde \tau}$ is $\F^{\P}_t$-measurable, and $\int^t_s Z_u e^{- \Gamma_u} \ud \Gamma_u$ is $\F^{\P}_t$-measurable as well.
\end{proof}

\section{Superhedging for payment streams}\label{sec: superhedging}

We now study the problem of superheging payment streams under model uncertainty. We stress that the dynamic superhedging problem in continuous time for payment streams has been not yet defined in the literature. Even in the case with a single prior, the problem is addressed only in discrete time, see e.g. \cite{Foll-Sch}, \cite{Penn} and \cite{Penn-dual}. Here we aim to fill this gap, by formulating rigorously the meaning of dynamic superhedging payment streams in continuous time and by analysing in detail its consequence.
A finite time horizon $[0,T]$ with $T > 0$ is fixed through out this section. 

\subsection{Optional decomposition}

We recall first some preliminary results of Section 2 in \cite{Nut-rob}, which are useful for further discussion.  Definitions and theorems in this section are all independent of the choice of the measurable space $\Omega$, the filtration $\Fb$ and the probability family $\P$.
In the sequel ``sigma martingale'' can be replaced by ``local martingale''. 

Let $S:=(S_t)_{t \in [0,T]}$ be an $m$-dimensional $\Fb$-adapted process with càdlàg paths, where $m$ is a positive integer. If under a probability $P$ the process $S$ is a $(P, \Fb)$-semimartingale, we denote its characteristics by $(B^P,C^P,\nu^P)$. By Proposition 2.2 of \cite{Neu-Nut-mea}, the process $S$ is also a $(P,\Fb^P_+)$-semimartingale with the same characteristics. Moreover, if $S$ is a $(P, \Fb)$-semimartingale for all $P \in \P$, we denote by $L(S, \P)$ the set of all $m$-dimensional $\Fb$-predictable processes which are $S$-integrable for all $P \in \P$, and by $ \prescript{(P)}{}{\!\int} \delta \ud S := (\prescript{(P)}{}{\!\int_0^t }\delta \ud S )_{t \in [0,T]}$ the usual It\^{o} integral under $P$ for all $\delta \in L(S, \P)$.

\begin{assump}\label{assump: saturated set of sigma mart meas}
	The following conditions hold:
	\begin{enumerate}
		\item $\P$ is a set of sigma martingale measures for $S$: the process $S$ is a $(P, \Fb^P_+)$-sigma martingale for all $P \in \P$;
		\item $\P$ is saturated: all equivalent sigma martingale measures of its element still belong to $\P$;
		\item $S$ has dominating diffusion under every $P \in \P$: we have $\nu^P \ll (C^P)^{ii}$ $P$-a.s. for all $i=1,...,m$ and for all $P \in \P$.
	\end{enumerate}
\end{assump}

\begin{rem}\label{rem: S continuous}
	If $S$ has continuous paths, then it always has dominating diffusion under a sigma martingale measure $P$, since its characteristics are reduced to $(0,C^P,0)$; in particular, it is a continuous local martingale under $P$.
\end{rem}

\begin{rem}
	Under the choice of $m = d$ and $S = B$, Lemma 4.2 and Proposition 4.3 of \cite{Nut-rob} give a sufficient condition such that Assumption \ref{assump: conditions on P (only one) for conditional sub exp} and Assumption \ref{assump: saturated set of sigma mart meas} are both satisfied.
\end{rem}

We recall Theorem 2.4 of \cite{Nut-rob}.

\begin{theorem}\label{theorem: opt decomposition}
	Under Assumption \ref{assump: saturated set of sigma mart meas}, let $Y:=(Y_t)_{t \in [0,T]}$ be a real-valued, $\Fb$-adapted process with càdlàg paths, which is a $(P, \Fb^P_+)$-local supermartingale for all $P \in \P$. Then there exists an $\Fb$-predictable process $\delta:=(\delta_t)_{t \in [0,T]}$ in $L(S, \P)$ such that
	\[ 
		Y - Y_0 - \prescript{(P)}{}{\!\!\!\!\int} \delta \ud S \text{ is nonincreasing } P\text{-a.s. for all } P \in \P.
 	\]
\end{theorem}

\subsection{Problem formulation}

We give now the formulation of the superhedging problem. Definitions in this section are independent of the choice of the measurable space $\Omega$, the filtration $\Fb$ and the probability family $\P$ as well.

We define the filtration $\Fb^{\P} := (\F^{\P}_t)_{t \in [0,T]}$ by
\[
	\F^{\P}_t := \F^*_t \vee \mathcal{N}_T^{\P}, \ \ \ \  t \in [0,T],
\]
where $\mathcal{N}_T^{\P}$ is the collection of sets which are $(P, \F_T)$-null for all $P \in \P$.
Let $A:=(A_t)_{t \in [0,T]}$ be a nonnegative $\Fb^{\P}$-adapted process with nondecreasing paths such that $A_t(\omega)$, $\omega \in \Omega$, is upper semianalytic for all $t \geqslant 0$. Without loss of generality we assume $A_0 = 0$. Let $S:=(S_t)_{t \in [0,T]}$ be an $m$-dimensional $\Fb^{\P}$-adapted process with càdlàg paths, which is a $(P, \Fb^P)$-semimartingale for all $P \in \P$. The processes $A$ and $S$ represent respectively an (eventually discounted) \emph{cumulative payment stream} and (eventually discounted) tradable assets on the market. 

We denote by $L(S, \P)$ the set of all $m$-dimensional $\Fb^{\P}$-predictable processes which are $S$-integrable for all $P \in \P$ and define the following set of  \emph{admissible strategies}
\[
	\Delta := \left\{ \delta \in L(S,\P) : \prescript{(P)}{}{\!\!\!\!\int}{\delta} \ud S \text{ is a }(P, \Fb^P_+)\text{-supermartingale for all } P \in \P \right\}.
\]

\begin{defn}\label{def: global superhedging strategy}
	A process $\delta \in \Delta$ is called \emph{robust global superhedging strategy for a cumulative payment stream $A$} if there exists $v \in \R$ such that 
	\[
		v + \prescript{(P)}{}{\!\!\!\!\int^\tau_0}{\delta}_u \ud S_u \geqslant A_\tau \ \ \ \ P\text{-a.s. for all } P \in \P,
	\]
	for all $[0,T]$-valued $\Fb$-stopping time $\tau$.
	%Under a given probability $P \in \P$, the process $(v + \prescript{(P)}{}{\!\int^t_0}{\delta}_u \ud S_u)_{t \in [0,T]}$ is called a \emph{$P$-superhedging value process for $A$}.
\end{defn}

\begin{defn}\label{def: local superhedging strategy}
	Let $\sigma, \tau$ be two $[0,T]$-valued $\Fb$-stopping times such that $\sigma \leqslant \tau$. A process $\delta \in \Delta$ is called \emph{robust local superhedging strategy for a cumulative payment stream $A$ on the random interval $[\sigma, \tau]$} if there exists a real-valued $\F^{\P}_{\sigma}$-measurable function $v$ such that
	\[
		v + \prescript{(P)}{}{\!\!\!\!\int^{\sigma'}_\sigma}{\delta}_u \ud S_u \geqslant A_{\sigma'} - A_\sigma \ \ \ \ P\text{-a.s. for all } P \in \P,
	\]
	for all $[0,T]$-valued $\Fb$-stopping time $\sigma'$ with $\sigma \leqslant \sigma' \leqslant \tau$.
\end{defn}

\noindent We note that Definition \ref{def: local superhedging strategy} agrees with the definition of superhedging strategies given in e.g. \cite{Foll-Sch}, \cite{Penn} and \cite{Penn-dual} in discrete time and without model uncertainty. Furthermore, clearly an admissible strategy $\delta$ is a robust global superhedging strategy if and only if it is a robust local superhedging strategy on all random intervals in $[0,T]$. Similarly, we define global and local superhedging prices as follows.

\begin{defn}\label{def: global superhedging price}
	We call \emph{robust global superhedging price for $A$} the value $\pi_0^T \in \R$ such that 
	\begin{align}
		\pi_0^T = &
		{\inf} \left\{ v \in \R : \exists \delta \in \Delta \text{ such that for every $[0,T]$-valued $\Fb$-stopping time } \tau, \right. \nonumber \\
		&  \left. v + \prescript{(P)}{}{\!\!\!\!\int^\tau_0}{\delta}_u \ud S_u \geqslant A_\tau \ P\text{-a.s. for all } P \in \P \right\} \label{eq: first formulation price}.
	\end{align}
\end{defn}

\begin{defn}\label{def: local superhedging price}
	Let $\sigma, \tau$ be two $[0,T]$-valued $\Fb$-stopping times  such that $\sigma \leqslant \tau$. We call \emph{robust local superhedging price for $A$ over the random interval $[\sigma, \tau]$} a real-valued  $\F^{\P}_{\sigma}$-measurable function $\pi_\sigma^\tau$ such that 
	\begin{align}
		\pi_\sigma^\tau =&{\text{ess inf }^P} \left\{ v \text{ is } \F^{\P}_{\sigma}\text{-measurable} : \exists {\delta} \in \Delta \text{ such that for every $\Fb$-stopping time }\sigma' \right. \nonumber \\
		& \left. \text{with } \sigma \leqslant \sigma' \leqslant \tau, \ v + \prescript{(P)}{}{\!\!\!\!\int^{\sigma'}_\sigma}{\delta}_u \ud S_u \geqslant A_{\sigma'} - A_\sigma \  P\text{-a.s. for all } P \in \P \right\}\nonumber \\
		& \ P\text{-a.s. for all } P \in \P. \label{eq: second formulation price}
	\end{align}
\end{defn}

\noindent Definition \ref{def: local superhedging price} agrees with the definition of superhedging price (or superhedging premium) given in e.g. \cite{Foll-Sch}, \cite{Penn} and \cite{Penn-dual} in discrete time and without model uncertainty. We emphasize that the robust local superhedging price is unique only up to a set $N \in \mathcal{N}^{\P}$. 

We are mainly interested in the following two problems.
\begin{enumerate}
	\item Show the existence of robust global and local superhedging prices as defined in Definition \ref{def: global superhedging price} and Definition \ref{def: local superhedging price} and determine their value.
	
	\item Show the existence of global and local superhedging strategies for a payment stream associated to robust global and local superhedging prices. In particular, we call \emph{optimal superhedging strategies for $A$} a robust global superhedging strategy $\delta$ for $A$ such that, for all $[0,T]$-valued $\Fb$-stopping times $\sigma, \sigma', \tau$ with $\sigma \leqslant \sigma' \leqslant \tau$, we have 
	\[
		\pi_\sigma^\tau + \prescript{(P)}{}{\!\!\!\!\int^{\sigma'}_\sigma}{\delta}_u \ud S_u \geqslant A_{\sigma'} - A_\sigma \ \ \ \ P\text{-a.s. for all } P \in \P.
	\]
\end{enumerate}
The first issue is a \emph{pricing problem}. The robust global (or resp. local) superhedging price of $A$ can be indifferently interpreted as the minimal amount of money the company should keep in order to be able to pay out in the future, or as the minimal price the product should be sold. 
The second problem is a \emph{hedging problem}. 
We emphasize the importance of distinguishing robust global and local superhedging problems. Clearly, for products with single payoff such as European contingent claims, %(with $A = (X \Ind{t = T})_{t \in [0,T]}$, where $X$ is an upper semianalytic, $\F^{\P}_T$-measurable and nonnegative function) or American type claims (with $A = (R_\tau \Ind{t = \tau})_{t \in [0,T]}$, where $\tau$ is a $[0,T]$-valued $\Fb$-stopping time and $R:= (R_t)_{t \in [0,T]}$ is an $\Fb^{\P}$-predictable nonnegative process, with $R(t, \omega) := R_t(\omega)$, $(t, \omega) \in [0, T] \times \Omega$, upper semianalytic in $(t, \omega)$), 
only the global problem is relevant. However, in the case of a generic payment stream, investors may be interested in the superhedging problem over a particular time interval.

\subsection{Robust superhedging for payment streams}\label{sec: financial duality payment stream}

We now study the dynamic superhedging for payment streams in the standard setting, where we use notations of Section \ref{sec: space construction}.

The following theorem %does not give a direct answer to the problems stated above, but 
is an intermediate step. 

\begin{theorem}\label{theo: duality final payment}
	Let Assumptions \ref{assump: conditions on P (only one) for conditional sub exp} and \ref{assump: saturated set of sigma mart meas} hold. Let $\sigma, \tau$ be two $[0,T]$-valued $\Fb$-stopping times such that $\sigma \leqslant \tau$, and $A := (A_t)_{t \in [0,T]}$ be a cumulative payment stream with $\mathcal{E}(A_T) < \infty$. If there exists an $\Fb^{\P}$-adapted process $Y=(Y_t)_{t \in [0,T]}$ with càdlàg path, such that for all $t \in [0,T]$
	\[
		Y_t = \mathcal{E}_t (A_\tau) \ \ \ \ P\text{-a.s. for all } P \in \P,
	\]
	then we have the following equivalent dualities for every $P \in \P$:
	\begin{align}
		&\mathcal{E}_\sigma(A_\tau) \nonumber \\
		=& 
		{\text{ess inf }^P} \left\{ v \text{ is } \F^{\P}_{\sigma }\text{-measurable} : \exists {\delta} \in \Delta \text{ such that }
		v + \prescript{(P')}{}{\!\!\!\!\int^\tau_\sigma}{\delta}_u \ud S_u \geqslant A_\tau \right. \nonumber	\\
		&\left.P'\text{-a.s. for all } P' \in \P \right\} \ \ \ \ P\text{-a.s.} \label{eq: first formulation duality}\\
		=& {\text{ess inf }^{P}} \left\{ v \text{ is } \F^{\P}_{\sigma}\text{-measurable} : \exists {\delta} \in \Delta \text{ such that }
		v + \prescript{(P')}{}{\!\!\!\!\int^\tau_\sigma}{\delta}_u \ud S_u \geqslant A_\tau \right. \nonumber	\\
		&\left. P'\text{-a.s. for all } P' \in \P(\sigma; P)\right\} \ \ \ \ P\text{-a.s.}, \label{eq: first formulation duality 2} 
	\end{align}	
	and
	\begin{align}
		&\mathcal{E}_\sigma(A_\tau - A_\sigma) \nonumber \\
		=& 
		{\text{ess inf }^P} \left\{ v \text{ is } \F^{\P}_{\sigma}\text{-measurable} : \exists {\delta} \in \Delta \text{ such that }
		v + \prescript{(P')}{}{\!\!\!\!\int^\tau_\sigma}{\delta}_u \ud S_u \geqslant A_\tau - A_\sigma \right. \nonumber	\\
		&\left.P'\text{-a.s. for all } P' \in \P \right\} \ \ \ \ P\text{-a.s.} \label{eq: second formulation duality}\\
		=& {\text{ess inf }^{P}} \left\{ v \text{ is } \F^{\P}_{\sigma}\text{-measurable} : \exists {\delta} \in \Delta \text{ such that } v + \prescript{(P')}{}{\!\!\!\!\int^\tau_\sigma}{\delta}_u \ud S_u \geqslant A_\tau - A_\sigma \right. \nonumber \\
		&\left. P'\text{-a.s. for all } P \in \P(\sigma; P) \right\} \ \ \ \ P\text{-a.s.} \label{eq: second formulation duality 2}
	\end{align}
\end{theorem}

\begin{proof}
	The proof is based on Theorem \ref{theorem: opt decomposition} and is similar to Theorem 3.2 of \cite{Nut-rob} and Theorem 3.4 of \cite{Bia-Man-bubb} with minor changes. We refer to \cite{Zhang} for further details.
\end{proof}

%\begin{rem}\label{rem: reformulation deterministic times}
%	Remark 2.4 (ii) of \cite{Nut-Han} shows that if Assumption \ref{assump: conditions on P (only one) for conditional sub exp} holds only for deterministic times then so does Proposition \ref{prop: construction F-cond sub exp (only one P)}. From our proof it is clear that in such case Theorem \ref{theo: duality final payment} can be reformulated in terms of deterministic times as well.
%\end{rem}

Theorem \ref{theo: duality final payment} extends Theorem 3.4 of \cite{Bia-Man-bubb} to the case of payment streams and can be considered as a dynamic version of Theorem 3.2 of \cite{Nut-rob}. It includes also the static robust superhedging dualities in e.g. \cite{Pos-Roy}, \cite{Dol-son} and \cite{Bay-Hua}. We note that a priori the robust global superhedging price of $A$ as defined in Definition \ref{def: global superhedging price} is higher than $\mathcal{E}(A_T)$ and the robust local superhedging price of $A$ on the interval $[\sigma, \tau]$ as defined in Definition \ref{def: local superhedging price} is higher than $\mathcal{E}_{\sigma} (A_\tau - A_\sigma)$. However, in the following we will see that equality holds.

For all $[0,T]$-valued $\Fb$-stopping times $\sigma, \tau$ such that $\sigma \leqslant \tau$, we define the following set:
\begin{align*}
	\mathcal{C}^\tau_{\sigma} : = &  \left\{\delta \in \Delta: \mathcal{E}_{\sigma_1}(A_\tau) + \prescript{(P)}{}{\!\!\!\!\int^{\sigma_2} _{\sigma_1}}{\delta}_u \ud S_u \geqslant A_{\sigma_2} \ P\text{-a.s.} \ \text{for all $[0,T]$-valued}\right. \nonumber \\
	&\left.\text{$\Fb$-stopping times } \sigma_1,{\sigma_2} \text{ such that } \sigma \leqslant \sigma_1\leqslant{\sigma_2}\leqslant\tau, \ \text{ for all } P \in \P \right\}. \nonumber
\end{align*}
If $\sigma, \sigma', \tau, \tau'$ are $[0,T]$-valued $\Fb$-stopping times such that $\sigma \leqslant \sigma' \leqslant \tau \leqslant \tau'$, then it clearly holds by definition 
\begin{equation} \label{eq: inclusion C}
	\mathcal{C}^{T}_0 \subseteq \mathcal{C}^{\tau'}_\sigma \subseteq \mathcal{C}^\tau_{\sigma}  \subseteq \mathcal{C}^\tau_{\sigma'}.
\end{equation}
The following theorem solves both the pricing and hedging problem for a payment stream.

\begin{theorem}\label{theo: superhedging strategy}
	Under the same assumptions as in Theorem \ref{theo: duality final payment}, we have:
	\begin{enumerate}
		\item the set $\mathcal{C}^{T}_0$ is not empty;
		
		\item the robust global superhedging price of $A$ is given by $\mathcal{E}(A_T)$ and the robust local superhedging price of $A$ on the interval $[\sigma, \tau]$ is given by $\mathcal{E}_{\sigma} (A_\tau - A_\sigma)$;
		
		\item the infimum value in (\ref{eq: first formulation price}) and (\ref{eq: second formulation price}) is attained, i.e. optimal superhedging strategies exist.
	\end{enumerate}
\end{theorem}

\begin{proof}
	Since it holds that
	\begin{align}
		\mathcal{E}_\sigma(A_\tau) - A_\sigma &:= \underset{P' \in \mathcal{ P}(\sigma;  P)}{\text{ess sup}^{P}} \E^{{P'}}[A_\tau | \mathcal{F}_\sigma]- A_\sigma = \underset{P' \in \mathcal{ P}(\sigma;  P)}{\text{ess sup}^{P}} \E^{{P'}}[A_\tau - A_\sigma | \mathcal{F}_\sigma] \nonumber \\
		&= \mathcal{E}_\sigma(A_\tau - A_\sigma)  \ \ \ \ P\text{-a.s. for all } P \in \P,
	\end{align}
	every set $\mathcal{C}^\tau_{\sigma}$ can be equivalently represented as
	\begin{align*}
		\mathcal{C}^\tau_{\sigma}
		= & \left\{\delta \in \Delta: \mathcal{E}_{\sigma_1}(A_\tau - A_{\sigma_1}) + \prescript{(P)}{}{\!\!\!\!\int^{\sigma_2} _{\sigma_1}}{\delta}_u \ud S_u \geqslant A_{\sigma_2} - A_{\sigma_1}  \ P\text{-a.s. for all $[0,T]$-valued}\right. \nonumber \\
			&\left. \ \text{$\Fb$-stopping times } \sigma_1,{\sigma_2} \text{ such that } \sigma \leqslant \sigma_1\leqslant{\sigma_2}\leqslant\tau, \text{ for all } P \in \P \right\}. 
	\end{align*}
	Hence, point 2 and point 3 follow from point 1 together with dualities (\ref{eq: first formulation duality}), (\ref{eq: second formulation duality}) and inclusion (\ref{eq: inclusion C}).
	
	Now we show the first point. Similar to the proof of Theorem \ref{theo: duality final payment} and Theorem 2.3 of \cite{Nut-Son}, by applying Theorem \ref{theorem: opt decomposition} it is possible to find an $\Fb^{\P}$-predictable process $\delta \in L(S, \P)$ such that for every $[0,T]$-valued $\Fb$-stopping time $\sigma$ we have
	\[
		\mathcal{E}_{\sigma}(A_T) + \prescript{(P)}{}{\!\!\!\!\int^{T}_\sigma} \delta_u \ud S_u \geqslant  A_T \ \ \ \  P\text{-a.s. for all } P \in \P.
	\] 
	In particular, if $\sigma'$ is another $[0,T]$-valued $\Fb$-stopping time such that $\sigma \leqslant \sigma'$, then
	\[
		\mathcal{E}_{\sigma}(A_T) + \prescript{(P)}{}{\!\!\!\!\int^{\sigma'}_{\sigma}} \delta_u \ud S_u + \prescript{(P)}{}{\!\!\!\!\int^{T}_{\sigma'}} \delta_u \ud S_u \geqslant  A_T \ \ \ \ P\text{-a.s. for all } P \in \P.
	\]
	Since $\prescript{(P)}{}{\!\int}{\delta} \ud S$ is a $(P,\Fb^P_+)$-supermartingale, by applying conditional expectation on both hand sides we get
	\[
		\mathcal{E}_{\sigma}(A_T) + \prescript{(P)}{}{\!\!\!\!\int^{\sigma'}_{\sigma}} \delta_u \ud S_u \geqslant E^{P}[ A_T| \F^P_{\sigma'+}] \ \ \ \ P\text{-a.s. for all } P \in \P.
	\]
	We note that since $A$ is nondecreasing, we have
	\[
		E^{P}[ A_T| \F^P_{\sigma'+}] - A_{\sigma'} =  E^{P}[ A_T  - A_{\sigma'}| \F^P_{\sigma'+}] \geqslant 0 \ \ \ \ P\text{-a.s. for all } P \in \P.
	\]
	Hence
	\[
		\mathcal{E}_{\sigma}(A_T) + \prescript{(P)}{}{\!\!\!\!\int^{\sigma'}_{\sigma}} \delta_u \ud S_u\geqslant A_{\sigma'} \ \ \ \ P\text{-a.s. for all } P \in \P.
	\]
	This shows that the set $\mathcal{C}^T_0$ is not empty.
\end{proof}

We stress that Theorem \ref{theo: duality final payment} and Theorem \ref{theo: superhedging strategy} can be carried out without changes also in the situation without model uncertainty, i.e. when we have a single prior $P$ which is a sigma (or local) martingale measure for $S$. %Apparently if a family of probability measures $\P$ satisfies Assumption \ref{assump: conditions on P (only one) for conditional sub exp} then it has automatically infinite elements, hence Corollary \ref{cor: conditional expectation decreasing} does not hold in the case of a single prior $P$. However, in the proof of Theorem \ref{theo: duality final payment} we only use Corollary \ref{cor: conditional expectation decreasing} to apply Proposition \ref{prop: opt decomposition}. This step can be done by directly in the case of one probability $P$ since $(E^P[A_T|\F_t])_{t \in [0,T]}$ is a $(P, \Fb)$-martingale. 
	
%We note furthermore that the lower bound of $P$-robust superhedging value processes $(\mathcal{E}_t(A_T)_{t \in [0,T]})$ for $P \in \P$ is not in general attainable, i.e. it is not in general possible to have a process $\delta \in \Delta$ such that
%\[
%	\mathcal{E}_t(A_T) = \mathcal{E}(A_T)  + \prescript{(P)}{}{\!\!\!\!\int^t_0}{\delta} \ud S \ \ \ \ \text{for all } t \in [0,T] \ P\text{-a.s. for all } P \in \P,
%\]
%unless the financial market is complete under every $P \in \P$.

\subsection{Robust superhedging in the reduced-form framework}\label{sec: superhedging reduced form}

In view of the construction in Section \ref{sec: space construction} and Section \ref{sec: construction G-conditional expectation}, we can now extend the superhedging results to the reduced-form setting.

Similar to Section \ref{sec: financial duality payment stream}, we define the filtration $\Gb^{\tilde{\P}} := (\G^{\tilde{\P}}_t)_{t \in [0,T]}$ by 
\[
	\G^{\tilde{\P}}_t := \G^*_t \vee \mathcal{N}_T^{\tilde{\P}}, \ \ \ \  t \in [0,T],
\]
where $\mathcal{N}_T^{\tilde{\P}}$ is the collection of sets which are $(\tilde P, \G_T)$-null for all $\tilde P \in \tilde{\P}$. Let $\tilde{A}:=(\tilde{A}_t)_{t \in [0,T]}$ be a nonnegative $\Gb^{\tilde{\P}}$-adapted process with nondecreasing paths,  such that $\tilde A_t$ is upper semianalytic for all $t \in [0,T]$ and $\tilde A_0 = 0$.
The process $\tilde A$ represents an (eventually discounted) cumulative payment stream on the extended market. We set $S$ to be an $m$-dimensional $\Gb^{\tilde{\P}}$-adapted process with càdlàg paths, which is a $(\tilde P, \Gb^{\tilde P})$-semimartingale for all $\tilde P \in \tilde{\P}$ and  represents (eventually discounted) tradable assets on the enlarged market. Let $\tilde{L}(S, \tilde{\P})$ be the set of all $m$-dimensional $\Gb^{\tilde{\P}}$-predictable processes which are $S$-integrable for all $\tilde P \in \tilde{\P}$. We define the following set of admissible strategies on the extended market,
\[
	\tilde{\Delta} := \left\{ \tilde{\delta} \in \tilde{L}(S,\tilde{\P}) : \prescript{(\tilde P)}{}{\!\!\!\!\int}{\tilde{\delta}} \ud S \text{ is a }(\tilde P, \Gb^{\tilde P}_+)\text{-supermartingale for all } \tilde P \in \tilde{\P} \right\},
\] 
where $ \prescript{(\tilde P)}{}{\!\int} \tilde{\delta} \ud S := (\prescript{(\tilde P)}{}{\!\int_0^t }\tilde{\delta} \ud S )_{t \in [0,T]}$ is the usual It\^{o} integral under $\tilde{P}$. 
Robust global and local superhedging strategies, robust global and local superhedging prices and the sets $\tilde{\mathcal{C}}^t_{s}$ with $0 \leqslant s \leqslant t \leqslant T$ are defined correspondingly as in Section \ref{sec: financial duality payment stream}. 

\noindent Theorem \ref{theo: duality final payment extended market} and Theorem \ref{theo: superhedging strategy on G market} are analogue to Theorem \ref{theo: duality final payment} and Theorem \ref{theo: superhedging strategy} for the $\Fb$-filtration.

\begin{theorem}\label{theo: duality final payment extended market}
	Let Assumption \ref{assump: conditions on P (only one) for conditional sub exp} hold for the probability family $\P$, Assumption \ref{assump: saturated set of sigma mart meas} hold for $\tilde{\P}$, and $\tilde A:=({\tilde A}_t)_{t \in [0,T]}$ be a cumulative payment stream with $\tilde{\mathcal{E}}_t (\tilde{A}_T) \in L^1(\tilde \Omega)$ for all $t \in [0,T]$. If $t \in [0,T]$ and there exists a $\Gb^{\P}$-adapted process $\tilde Y=(\tilde Y_s)_{s \in [0,T]}$ with càdlàg paths, such that for $s \in [0,t]$
	\[
		\tilde Y_s = \tilde{\mathcal{E}}_s (\tilde A_{t}) \ \ \ \ \tilde P\text{-a.s. for all } \tilde P \in \tilde{\P},
	\]
	and if the tower property holds for $A_t$, i.e. for all $r,s \in [0,t]$ with $r \leqslant s$,
	\[
		\tilde{\mathcal{E}}_r (\tilde A_{t}) = \tilde{\mathcal{E}}_r (\tilde{\mathcal{E}}_s(\tilde A_{t})) \ \ \  \ \tilde P\text{-a.s. for all } \tilde P \in \tilde{\P},
	\]
	then we have the following equivalent dualities for all $\tilde P \in \tilde{\P}$ and $0 \leqslant s \leqslant t \leqslant T$:
	\begin{align*}
		&\mathcal{\tilde E}_{s}({\tilde A}_{t}) \nonumber \\
		=& {\text{ess inf }^{\tilde P}} \{ \tilde v \text{ is } \G^{\tilde{\P}}_{{s}}\text{-measurable} : \exists {\tilde{\delta}} \in \tilde{\Delta} \text{ such that }
		\tilde v + \prescript{(\tilde{P}')}{}{\!\!\!\!\int^{t}_s}{\tilde \delta}_u \ud S_u \geqslant {\tilde A}_{t} \  \tilde{P}'\text{-a.s.} \nonumber \\
		&\text{for all } \tilde{P}' \in \tilde{\P}\} \ \ \ \ \tilde P\text{-a.s.} \\
		=& {\text{ess inf }^{\tilde{P}}} \{ \tilde v \text{ is } \G^{\tilde{\P}}_{{s}}\text{-measurable} : \exists {\tilde{\delta}} \in \tilde{\Delta} \text{ such that }
		\tilde v + \prescript{(\tilde{P}')}{}{\!\!\!\!\int^{t}_s}{\tilde \delta}_u \ud S_u \geqslant {\tilde A}_{t} \ \tilde{P}'\text{-a.s.} \nonumber \\
		& \text{for all } \tilde{P}' \in \tilde{\P}(s;\tilde P)\} \ \ \ \ \tilde P\text{-a.s.},
	\end{align*}
	and
	\begin{align*}
		&\mathcal{\tilde E}_{s}({\tilde A}_{t} - {\tilde A}_{s})  \nonumber \\
		=& {\text{ess inf }^{\tilde P}} \{ \tilde v \text{ is } \G^{\tilde{\P}}_{{s}}\text{-measurable} : \exists {\tilde{\delta}} \in \tilde{\Delta} \text{ such that }
		\tilde v + \prescript{(\tilde{P}')}{}{\!\!\!\!\int^{t}_s}{\tilde \delta}_u \ud S_u \geqslant {\tilde A}_{t} - {\tilde A}_{s} \nonumber \\
		& \ \tilde{P}'\text{-a.s. for all } \tilde{P}' \in \tilde{\P}\} \ \ \ \ \tilde P\text{-a.s.} \\
		=& {\text{ess inf }^{\tilde{P}}} \{ \tilde v \text{ is } \G^{\tilde{\P}}_{{s}}\text{-measurable} : \exists {\tilde{\delta}} \in \tilde{\Delta} \text{ such that }\tilde v + \prescript{(\tilde{P}')}{}{\!\!\!\!\int^{t}_s}{\tilde \delta}_u \ud S_u \geqslant {\tilde A}_{t} - {\tilde A}_{s} \nonumber \\
		& \ \tilde{P}'\text{-a.s. for all } \tilde{P'} \in \tilde{\P}(s; \tilde P)\} \ \ \ \ \tilde P\text{-a.s.}
	\end{align*}
\end{theorem}

\begin{proof}
	The proof of the theorem is the same as in Theorem \ref{theo: duality final payment}. Indeed, we can apply Theorem \ref{theorem: opt decomposition} to the measurable space $\tilde \Omega$ with filtration $\Gb^{\tilde{\P}}$ and to the process $\tilde Y$.
\end{proof}

\begin{theorem}\label{theo: superhedging strategy on G market}
	Under the same assumptions in Theorem \ref{theo: duality final payment extended market}, for $0 \leqslant s \leqslant t \leqslant T$, we have the following statements.
	\begin{enumerate}
		\item The set $\mathcal{\tilde{C}}^{T}_0$ is not empty.
		
		\item The robust global superhedging price of $\tilde A$ is given by $\mathcal{\tilde{E}}(\tilde{A}_T)$ and the robust local superhedging price of $\tilde A$ on the interval $[s, t]$ is given by $\mathcal{\tilde{E}}_{s} (\tilde{A}_t - \tilde{A}_s)$.
		
		\item Optimal superhedging strategies exist.
		%\item For all $\tilde P \in \tilde{\P}$, the set of $\tilde P$-robust superhedging value processes for $\tilde A$ are $\tilde P$-a.s. bounded from below by $(\mathcal{\tilde{E}}_t(\tilde{A}_T))_{t \in [0,T]}$.
	\end{enumerate}
\end{theorem}

\begin{proof}
	The theorem can be proved in the same way as in Theorem \ref{theo: superhedging strategy}.
\end{proof}

By using the results in Section \ref{sec: main products}, we show that the superhedging problem can be solved for all main credit and insurance cash flows.
As already noticed in e.g. \cite{Bar}, \cite{Bia-Rhe-r} and \cite{Bia-Zha}, we recall that the three kinds of main products are special cases of payment streams by setting
\begin{equation}\label{eq: stream product 1}
	\tilde A_t = \Ind{\tilde \tau > T} Y \Ind{t = T}, \ \ \ \  t \in [0,T],
\end{equation}
\begin{equation}\label{eq: stream product 2}
	\tilde{A}_0 = 0, \ \ \ \ \tilde A_t = \Ind{0 < \tilde \tau \leqslant t} Z_{\tilde \tau}, \ \ \ \  t \in [0,T],
\end{equation}
or 
\begin{equation}\label{eq: stream product 3}
	\tilde{A}_0 = 0, \ \ \ \ \tilde A_t = \Ind{0 < \tilde \tau \leqslant t} C_{\tilde \tau} + \Ind{\tilde \tau > t} C_t, \ \ \ \  t \in [0,T],
\end{equation}
respectively. %Furthermore, in these cases $(\tilde{\P}, \Gb)$-conditional expectation can be expressed more explicitly. 

\begin{prop}\label{prop: measurability main products}
	Under the same assumptions of Lemma \ref{lemma: representation first product}, Corollary \ref{cor: construction G-cond sub exp} and Corollary \ref{cor: third product}, if in addition the family $\P$ is tight and $\mu$, $Y$, $Z$ and $C$ are bounded and continuous in $\omega$ $P$-a.e. for all $P \in \P$, then the processes
	\[
		\left( \mathcal{\tilde E}_t \left( \Ind{\tilde \tau > T}Y \right) \right)_{t \in [0,T]}, \ \ \ \ \left(\mathcal{\tilde E}_t \left( \Ind{0 < \tilde \tau \leqslant T} Z_{\tilde \tau} \right)\right)_{t \in [0,T]},
	\]
	and
	\[
		\left(\mathcal{\tilde E}_t \left( \int^T_0(1-H_u) \ud C_u \right)\right)_{t \in [0,T]}
	\]
	are $\Gb^*$-adapted and respectively equal to a càdlàg process $Y:=(Y_t)_{t \in [0,T]}$  $\tilde P$-a.s. for all $\tilde P \in \tilde{\P}$. 
\end{prop}

\begin{proof}
	The three processes are clearly $\Gb^*$-adapted by definition. For every $t \in [0,T]$, by Lemma \ref{lemma: representation first product}, Corollary \ref{cor: construction G-cond sub exp} and Corollary \ref{cor: third product}, we have
	\begin{align*}
		\mathcal{\tilde E}_t \left( \Ind{\tilde \tau > T}Y \right) &= \Ind{\tilde \tau > t} \mathcal{E}_t \left( Y e^{- \int^T_t \mu_u \ud u} \right)\\
		&= \Ind{\tilde \tau > t}  e^{\int^t_0 \mu_u \ud u} \mathcal{E}_t\left( Y e^{- \int^T_0 \mu_u \ud u} \right) \ \ \ \ \tilde P\text{-a.s. for all } \tilde P \in \tilde {\P},
	\end{align*}
	\begin{align*}
		\mathcal{\tilde E}_t \left( \Ind{0 < \tilde \tau \leqslant T} Z_{\tilde \tau} \right) =& \ \mathcal{\tilde E}_t \left( \Ind{t < \tilde \tau \leqslant T} Z_{\tilde \tau} \right) + \Ind{0 < \tilde \tau \leqslant t} Z_{\tilde \tau} \\
		=& \ \Ind{\tilde \tau > t} \mathcal{E}_t \left( \int^T_t Z_u e^{- \int^u_t \mu_v \ud v} \mu_u \ud u \right) + \Ind{0 < \tilde \tau \leqslant t} Z_{\tilde \tau} \\
		=& \ \Ind{\tilde \tau > t} e^{\int^t_0 \mu_v \ud v} \left[  \mathcal{E}_t \left( \int^T_0 Z_u e^{- \int^u_0 \mu_v \ud v} \mu_u \ud u \right) -  \int^t_0 Z_u e^{- \int^u_0 \mu_v \ud v} \mu_u \ud u\right]\\
		&+ \Ind{0 < \tilde \tau \leqslant t} Z_{\tilde \tau} \ \ \ \ \tilde P\text{-a.s. for all } \tilde P \in \tilde {\P},
	\end{align*}
	and
	\begin{align*}
		&\mathcal{\tilde E}_t \left( \int^T_0(1-H_u) \ud C_u \right)\\
		=& \ \mathcal{\tilde E}_t \left( \int^T_t(1-H_u) \ud C_u \right) - \int^t_0 (1-H_u) \ud C_u\\
		=& \ \Ind{\tilde \tau > t} \mathcal{E}_t \left( \int^T_t C_{u} e^{- \int^u_t \mu_v \ud v} \mu_u \ud u + C_T e^{- \int_t^T \mu_u \ud u} \right) - \left(\Ind{0 < \tilde \tau \leqslant t }C_{\tilde \tau} + \Ind{\tilde \tau > t}C_t\right)\\
		=& \ \Ind{\tilde \tau > t} e^{\int^t_0 \mu_v \ud v} \left[ \mathcal{E}_t \left( \int^T_0 C_{u} e^{- \int^u_0 \mu_v \ud v} \mu_u \ud u + C_T e^{- \int_0^T \mu_u \ud u}  \right)  - \int^t_0 C_{u} e^{- \int^u_0 \mu_v \ud v} \mu_u \ud u \right]\\
		&- \left(\Ind{0 < \tilde \tau \leqslant t }C_{\tilde \tau} + \Ind{\tilde \tau > t}C_t\right)\ \ \ \ \tilde P\text{-a.s. for all } \tilde P \in \tilde {\P}.
	\end{align*}
	Under our assumptions, Proposition \ref{prop: conditional expectation cadlag} shows that
	\[
		\left(\mathcal{E}_t \left( Y e^{- \int^T_0 \mu_u \ud u} \right)\right)_{t \in [0,T]}, \ \ \ \ \left( \mathcal{E}_t \left( \int^T_0 Z_u e^{- \int^u_0 \mu_v \ud v} \mu_u \ud u \right)\right)_{t \in [0,T]},
	\]
	and
	\[
		\left(\mathcal{E}_t \left( \int^T_0 C_{u-} e^{- \int^u_0 \mu_v \ud v} \mu_u \ud u + C_T e^{- \int_0^T \mu_u \ud u} \right)\right)_{t \in [0,T]}
	\]
	are càdlàg, hence the thesis follows.
\end{proof}

As a consequence, we now show that the superhedging price and strategy can be determined for the credit or insurance products of the form (\ref{eq: stream product 1}), (\ref{eq: stream product 2}) and (\ref{eq: stream product 3}).

\begin{cor}
	Under the same assumptions of Proposition \ref{prop: measurability main products} and Proposition \ref{prop: tower property main products}, if in addition $\P$ satisfies Assumption \ref{assump: saturated set of sigma mart meas}, then Theorem \ref{theo: duality final payment extended market} and Theorem \ref{theo: superhedging strategy on G market} apply to credit or insurance products of the form (\ref{eq: stream product 1}), (\ref{eq: stream product 2}) and (\ref{eq: stream product 3}).
\end{cor}

\begin{proof}
	It follows directly from Proposition \ref{prop: tower property main products}, Proposition \ref{prop: measurability main products},  Theorem \ref{theo: duality final payment extended market} and Theorem \ref{theo: superhedging strategy on G market}.
\end{proof}

%\begin{rem}\label{rem: third product}
%	We note that Corollary \ref{cor: measurability main products} can be applied also to the third type of main products which is the sum of the first two. While in general $(\tilde{\P}, \Gb)$-conditional expectation defined in (\ref{eq: existence G conditional expectation}) is only sublinear and not linear, in this special case, $(\tilde{\P}, \Gb)$-conditional expectation of the sum is actually the sum of $(\tilde{\P}, \Gb)$-conditional expectations by definition (\ref{eq: existence G conditional expectation}).
%	More precisely, if $\tilde A$ is defined by (\ref{eq: stream product 3}), then
%	\begin{align*}
%		\mathcal{\tilde E}_s ( \tilde A_t - \tilde A_s ) &= \mathcal{\tilde E}_s ( \Ind{s < \tilde \tau \leqslant t} C_{\tilde \tau-} + \Ind{\tilde \tau > T} C_T\Ind{t = T})\\
%		&= \mathcal{\tilde E}_s ( \Ind{s < \tilde \tau \leqslant t} C_{\tilde \tau-}) + \mathcal{\tilde E}_s (\Ind{\tilde \tau > T} C_T\Ind{t = T}).
%	\end{align*}
%	In particular,
%	\begin{align*}
%		\mathcal{\tilde E}_s ( \tilde A_t - \tilde A_s )= \Ind{\tilde \tau > s} \left(\mathcal{E}_s \left( \int^t_s A_u e^{- \int^u_s \mu_v \ud v} \mu_u \ud u \right) + \Ind{t = T} \mathcal{E}_s \left( A_T e^{- \int^T_s \mu_u \ud u} \right)\right)&\\
%		\tilde P\text{-a.s. for all } \tilde P \in \tilde{\P}&.
%	\end{align*}
%	with $C$ bounded and $P$-a.e. uniformly continuous on $\Omega$ for all $P \in \P$, then $\tilde{\mathcal{E}}_\cdot(\tilde{A}_t - \tilde{A}_\cdot)$ is a càdlàg process and for $0 \leqslant s \leqslant t \leqslant T$,
%\end{rem}

\section*{Acknowledgements}

The authors would like to thank Shige
Peng, Shandong University, for valuable discussions and the reference to the paper \cite{Yan}.

\appendix

\section{Counterexample for the tower property}\label{app: counterexample}

In this section we provide a counterexample to show that the classic tower property does not hold in general for the $(\tilde \P, \Gb)$-conditional expectation constructed in Section \ref{sec: construction G-conditional expectation}. 

Let $\Omega = C_0(\R_+, \R^d)$ and consider the $G$-conditional defined in e.g. \cite{Peng-non} as $(\P, \Fb)$-conditional expectation.
Since the $G$-conditional expectation  is only sublinear, there exist $t \geqslant 0$ and sufficiently regular functions $X$, $Y$ on $\Omega$ such that on a measurable set $A$ with $P(A) > 0$ for all $P \in \P$, the following strict inequality holds
\begin{equation}\label{eq: strict inequality}
	\mathcal{E}_t(X)(\omega) + \mathcal{E}_t(Y)(\omega) > \mathcal{E}_t(X + Y)(\omega) \ \ \ \ \text{for all } \omega \in A.
\end{equation}
Then there exists $s$ with $s < t$ such that
\begin{equation}\label{eq: strict inequality with s}
	\mathcal{E}_s ( \mathcal{E}_t(X) + \mathcal{E}_t(Y)) > \mathcal{E}_s (\mathcal{E}_t(X + Y)) \ \ \ \ P\text{-a.s. for all } P \in \P \text{ on } A.
\end{equation}
Indeed, if there exists a measurable subset $B \subseteq A$ with $P(B) > 0$ for all $P \in \P$, such that for all $s < t$ we have
\begin{equation*}
	\mathcal{E}_s ( \mathcal{E}_t(X) + \mathcal{E}_t(Y)) = \mathcal{E}_s (\mathcal{E}_t(X + Y)) \ \ \ \ P\text{-a.s. for all } P \in \P \text{ on } B,
\end{equation*}
then by taking the limit for $s \uparrow t$, we get
\begin{equation*}
	\mathcal{E}_t ( \mathcal{E}_t(X) + \mathcal{E}_t(Y)) = \mathcal{E}_t (\mathcal{E}_t(X + Y)) \ \ \ \ P\text{-a.s. for all } P \in \P \text{ on } B,
\end{equation*}
since the operator $\mathcal{E}_{t}$ is continuous in $t$ in the case of the $G$-conditional expectation, see e.g. \cite{Son-Tou} and \cite{Song}.
By (\ref{eq: consist cond for t}), the above equality is equivalent to 
\begin{equation*}
	\mathcal{E}_t(X) + \mathcal{E}_t(Y) = \mathcal{E}_t(X + Y) \ \ \ \ P\text{-a.s. for all } P \in \P \text{ on } B,
\end{equation*}
which contradicts (\ref{eq: strict inequality}).\\
Now we take $r, l$ with $s < r \leqslant t \leqslant l$ and define
\[
	\bar X := \frac{X}{e^{-\Gamma_s} - e^{-\Gamma_r}}, \ \ \ \ \bar Y := \frac{Y}{e^{-\Gamma_l}}.
\]
Inequality (\ref{eq: strict inequality with s}) thus equals the following
\begin{align}\label{eq: strict inequality example}
	& \ \mathcal{E}_s \left( ({e^{-\Gamma_s} - e^{-\Gamma_r}}) \ \mathcal{E}_t(\bar X) + \mathcal{E}_t(e^{-\Gamma_l} \bar Y) \right) \nonumber \\
	>& \ \mathcal{E}_s \left (\mathcal{E}_t(({e^{-\Gamma_s} - e^{-\Gamma_r}}) \bar X + e^{-\Gamma_l} \bar Y) \right) \ \ \ \ P\text{-a.s. for all } P \in \P \text{ on } A.
\end{align}
If we set 
\[
	\tilde X := \Ind{\tilde \tau \leqslant r} \bar X + \Ind{\tilde \tau > l} \bar Y,
\]
then the classic tower property does not hold for $\tilde X$, since 
 (\ref{eq: second inequality}) in Theorem \ref{theo: tower property} becomes a strict inequality on $A$.
 
\section{Sufficient conditions for the tower property}\label{app: other condition tower property}

In this section we state some other sufficient conditions which guarantees the tower property for $(\tilde \P, \tilde \Gb)$-conditional expectation. We note that these conditions do not include the case in Proposition \ref{prop: tower property main products}.

The following useful theory, called Yan's Commutability Theorem, can be found in \cite{Yan} and in Theorem a3 of \cite{Peng}. 

\begin{theorem}\label{theo: Yan}
	Let $(\Omega, \F, P)$ be an arbitrary probability space and $H$ be a subset of $L^1(\Omega, \F, P)$ such that $\sup_{\xi \in H} E^P[\xi] < + \infty$. The following statements are equivalent.
	\begin{enumerate}
		\item For all $\varepsilon > 0$ and $\xi_1$, $\xi_2 \in H$, there exists a $\xi_3 \in H$ such that
		\[
			E^P[(\xi_1 \vee \xi_2 - \xi_3)^+] \leqslant \varepsilon.
		\] 
		\item $E^P[ \underset{\xi \in H}{\text{ess sup}^{P}} \xi] = \underset{\xi \in H}{\text{sup}} E^P[\xi].$
		\item For any sub-$\sigma$-algebra $\mathcal{J}$ of $\F$, we have 
		\[
			E^P\left. \left[\underset{\xi \in H}{\text{ess sup}^{ P}} \xi\right| \mathcal{J} \right] = \underset{\xi \in H}{\text{ess sup}^{P}} E^P\left. \left[\xi \right| \mathcal{J} \right].
		\]
	\end{enumerate}
\end{theorem}

\begin{prop}
	Under the same assumptions of Theorem \ref{theo: tower property} or Corollary \ref{cor: tower property}, the tower property holds for $(\tilde \P, \tilde \Gb)$-conditional expectation, i.e.
	\begin{equation}
		\mathcal{\tilde E}_{s}(\mathcal{\tilde E}_{t}(\tilde X)) = \mathcal{\tilde E}_{s}(\tilde X) \ \ \ \ \tilde P \text{-a.s. for all } \tilde P \in \tilde{\P},
	\end{equation}
	with $0 \leqslant s \leqslant t$,
	if one of the following conditions is satisfied
	\begin{enumerate}
		\item $\tilde X$ does not depend on $\hat \omega \in \hat \Omega$;
		
		\item $\mathcal{E}_t (\Ind{\tilde \tau > s}\tilde X )$ is $\mathcal{B}(\hat \Omega)$-measurable and		
		$\mathcal{E}_t (E^{\hat P} [\Ind{\tilde \tau > s}\tilde X] ) = E^{\hat P} [\mathcal{E}_t (\Ind{\tilde \tau > s}\tilde X )]$ $P$-a.s. for all $P \in \P$ and for all $0 \leqslant s \leqslant t$;
		
		\item for all $P \in \P$, $P$-a.e. $\omega$, $0 \leqslant s \leqslant t$, $\varepsilon > 0$ and $P_1, P_2 \in \P$, there is a $P_3 \in \P$ such that if $\tilde Y := \Ind{\tilde \tau > s}\tilde X$, the functions 
		\[
			\xi_i(\hat \omega) = \int_\Omega \tilde Y (\omega \otimes_{t} \omega', \hat{\omega}) \ud P_i(\omega'), \ \ \ \  i = 1, 2, 3,
		\]
		with $\omega \otimes_{t} \omega'$ defined in (\ref{eq: pasting})		are $\mathcal{B}(\hat \Omega)$-measurable and
		\[
			E^{\hat P} [ (\xi_1 \vee \xi_2 - \xi_3)^+] \leqslant \varepsilon \ \ \ \ P\text{-a.s.}
		\]
\end{enumerate}	
\end{prop}

\begin{proof}
	Condition 1 is trivial. Indeed, by point 1 of Remark \ref{rem: properties G conditional expectation}, in such case the $(\tilde \P, \tilde \Gb)$-conditional expectation is reduced to the $(\P, \tilde \Fb)$-conditional expectation which satisfies the tower property.
	
	If condition 2 is satisfied, according to the proof of Theorem \ref{theo: tower property}, it is sufficient to check that (\ref{eq: first inequality}) and (\ref{eq: second inequality}) are equalities. We have indeed
	\begin{align*}
		&e^{\Gamma_s}\mathcal{ E}_{s}\left( E^{\hat{P}}[ \mathcal{E}_t (\Ind{s < \tilde \tau \leqslant t}\tilde X)] + \mathcal{ E}_{t}( E^{\hat{P}}[ \Ind{\tilde \tau > t}\tilde X] ) \right) \\
		= &e^{\Gamma_s}\mathcal{ E}_{s}\left( E^{\hat{P}}[ \Ind{s < \tilde \tau \leqslant t} \mathcal{E}_t (\tilde X)] + E^{\hat{P}}[ \Ind{\tilde \tau > t} \mathcal{ E}_{t}(  \tilde X )] \right) \\
		= &e^{\Gamma_s}\mathcal{ E}_{s}\left( E^{\hat{P}}[ \Ind{s < \tilde \tau \leqslant t} \mathcal{E}_t (\tilde X) + \Ind{\tilde \tau > t} \mathcal{ E}_{t}(  \tilde X )] \right) \\
		= &e^{\Gamma_s}\mathcal{ E}_{s}\left( E^{\hat{P}}[ \Ind{ \tilde \tau > s} \mathcal{E}_t (\tilde X)] \right)\\
		= &e^{\Gamma_s}  \mathcal{ E}_{s}( \mathcal{E}_t ( E^{\hat{P}}[\Ind{ \tilde \tau > s} \tilde X ]) )\\
		= &e^{\Gamma_s} \mathcal{ E}_{s}(  E^{\hat{P}}[ \Ind{ \tilde \tau > s} \tilde X ] ) \ \ \ \ P \text{-a.s. for all } P \in \P. 
	\end{align*}
	
	Condition 3 is equivalent to condition 2 by using the equivalence between statements 1 and 2 in Yan's Commutability Theorem \ref{theo: Yan}.
\end{proof}

\bibliographystyle{plain}
%\nocite{*}
\bibliography{bibl2}

\end{document}